\documentclass[12pt]{article}

\usepackage{amsfonts,amsthm, longtable, booktabs, algorithm, algpseudocode,geometry, hyperref, multirow}
\usepackage[authoryear,comma]{natbib}
\usepackage[tbtags]{amsmath}
\usepackage{color}
\hypersetup{citecolor=blue,colorlinks=true}

\RequirePackage[pdftex]{graphicx}
\DeclareGraphicsExtensions{.eps,.pdf,.jpeg,.png}

\RequirePackage{fancyhdr}
\usepackage{autonum}

\numberwithin{equation}{section}

\theoremstyle{plain}
\newtheorem{theorem}{Theorem}

\theoremstyle{definition}

\theoremstyle{remark}
\newtheorem*{remark}{Remark}

\renewcommand{\boldsymbol}[1]{{#1}} 

\newcommand{\X}{{\boldsymbol{X}}}
\newcommand{\bsigma}{{\boldsymbol{\Sigma}}}

\newcommand{\bO}{{\boldsymbol{\Omega}}}
\newcommand{\bI}{{\boldsymbol{I}}}
\newcommand{\bJ}{{\boldsymbol{J}}}
\newcommand{\bK}{{\boldsymbol{K}}}

\DeclareMathOperator{\tr}{tr}
\newcommand{\be}{\begin{equation}}
	\newcommand{\ee}{\end{equation}}
\newcommand{\ba}{\begin{eqnarray}}
	\newcommand{\ea}{\end{eqnarray}}
\newcommand{\bsp}{\begin{split}}
	\newcommand{\esp}{\end{split}}

\title{Robust Density Power Divergence Estimates for Panel Data Models}

\author{Abhijit Mandal\\
	Department of Mathematical Sciences, Univeristy of Texas at El Paso, El Paso, U.S.A.  \\
	Beste Hamiye Beyaztas\\
	Department of Statistics, Istanbul Medeniyet University, Istanbul, Turkey\\
	Soutir Bandyopadhyay\\
	Department of Applied Mathematics, Statistics
	Colorado School of Mines, Denver, U.S.A. }

\begin{document}
	
	\maketitle

	\begin{abstract}
		The panel data regression models have become one of the most widely applied statistical approaches in different fields of research, including social, behavioral, environmental sciences, and econometrics. However, traditional least-squares-based techniques frequently used for panel data models are vulnerable to the adverse effects of the data contamination or outlying observations that may result in  biased and inefficient estimates and misleading statistical inference. In this study, we propose a {\it minimum density power divergence} estimation procedure for panel data regression models with random effects to achieve robustness against outliers. The robustness, as well as the asymptotic properties of the proposed estimator, are rigorously established.  The finite-sample properties of the proposed method are investigated through an extensive simulation study and an application to climate data in Oman. Our results demonstrate that the proposed estimator exhibits improved performance over some traditional and robust methods in the presence of data contamination.
	\end{abstract}
	
	\begin{keywords}
		Minimum density power divergence; Panel data; Random effects; Robust estimation.
	\end{keywords}

	\section{Introduction}\label{sec:1}

	The advancements in applied and methodological researches on panel data have been growing remarkably since the seminal paper of \cite{Balestra1966}. Panel data, sometimes referred to as longitudinal data, is multi-dimensional data consisting of observations collected over a period of time on the same set of cross-sectional units. This multi-dimensionality provides more information than either pure cross-sectional or pure time-series data since its grouping structure allows to control the individual-specific heterogeneity and the intra-individual dynamics  \citep{Baltagi2005, Hsiao1985}. Some of the advantages of using this type of data include more variability, more degrees of freedom, more efficiency and less collinearity between the variables, and better capability to identify and measure the effects which are not entirely detectable in a single cross-sectional or time-series data \citep{Baltagi2005}. Moreover, one of its attractive features is that more informative data results in reliable statistical inference by improving the accuracy and precision in estimating model parameters \citep{Beyaztas2020}.

	The main factors affecting this excessive growth in panel data studies can be summarized as availability of panel data, greater ability for reflecting the complexity of human behavior than a single cross-sectional or time-series data, and challenging methodology  \citep{Hsiao2007}. In general, collecting panel data is more problematic than collecting the cross-sectional or time-series data since such data consists of a large number of observations. However,  with the latest advances of information technology and systems in different fields of applications, generating and storing high throughput data have become much easier in recent times.

	Due to the systematic differences across cross-sectional units, all pertinent information may not always be captured in regressions using aggregated time-series and pure cross-section data, and omission of such information may further lead to biased and inconsistent statistical inference results
	\citep{Jirata2014}. To this end, panel data models have become one of the cornerstone approaches in empirical research in fields of economics, social sciences, and medical sciences because these models allow to explain the individual behavior over time while capturing the inter-temporal dynamics. One of the most conspicuous among several attractive features of panel data models is in their capability to account for the unobserved individual-specific heterogeneity. In this context, the fixed and random effects models are the most commonly used panel data regression models including the individual-specific components. In the fixed-effects model, the unobserved heterogeneity across individuals captured by the time-invariant intercept terms is allowed to be correlated with the explanatory variables and assumed to be fixed. On the other hand, those individual effects are treated as a part of the disturbance term and controlled by the differences in the error variance components in the random-effects model. A key assumption that distinguishes these two principle models is that the individual-specific effects are assumed to be independently distributed of the explanatory variables in the random-effects model while fixed effects models allow for a limited form of endogeneity \citep{Mundlak1978}.
	For a more detailed exposition of the research on technical details of linear panel data models and its applications, see 
	\cite{Wallace1969,Maddala1973,Mundlak1978,Laird1982,Cox2002,Diggle2002,Fitzmaurice2004,Gardiner2009} and \cite{athey2021matrix}.

	Typically, panel data allows us to exploit different sources of variation: (i) variation within cross-sectional units (\textit{within variation}), (ii) variation between cross-sectional units (\textit{between variation}), and (iii) variation over both time and cross-sectional units (\textit{overall variation}), and the estimation techniques differ depending on the source of variation \citep{Kennedy2003}. For example, within estimator (also called \textit{fixed-effects estimator}) utilizes within variation while the generalized least squares (GLS) estimators  used for random effects models take into account both within and between variations. The majority of the regression techniques rely on using the least-squares (LS) based estimators for statistical inferences about the parameters of linear panel data models \citep{Aquaro2013}. However, obtaining consistent estimates of model parameters in traditional estimation techniques depends on some restrictive assumptions that may not be achieved in practice. Those assumptions such as normality and homoscedasticity of the error terms and strict exogeneity with respect to the error terms make the LS-based methods vulnerable to the adverse effects of the outliers and data contamination (\citealp{Greene2017,Kutner2004,Visek2015}). In panel data models, different types of outliers depending on the sources of contamination, i.e.,  vertical outliers, horizontal outliers (in the error term) and leverage points (in the explanatory variables), can arise because of the measurement error, typing error, transmission or copying error and naturally unusual data points (\citealp{Rousseeuw1990,Rousseeuw2003,Maronna2006,Bramati2007,Bakar2015}). Moreover, the outliers may occur in a block form (called as \textit{block-concentrated outliers}) such that most of the outlying observations tend to be concentrated within cross-sectional units, i.e., in a few time-series \citep{Bramati2007}. As a consequence, the LS-based approaches such as the ordinary least squares (OLS) estimator, GLS estimator, etc. may lead to substantial degradation of accuracy and erroneous estimates due to the sensitivity to outliers and/or any departure from the model assumptions. Furthermore, the outlying observations may not be determined by looking at the LS residuals or using standard outlier diagnostics due to the potential vulnerability of the complex nature of panel data to the masking effect. Although it is more crucial to have a robust method in the context of panel data models especially in presence of the contaminated datasets, most of the efforts have been devoted to the development of robust techniques for linear regression models and the existing literature for static panel data models cover only a few approaches. Some of those robust techniques as alternatives to the fixed effects estimator have been developed by \cite{Bramati2007} by considering high breakdown point of the well-known robust regression estimators, namely, least trimmed squares (LTS) estimator of \cite{Rousseeuw1984} and MS estimates of \cite{Maronna2000}. Also, a robust estimation procedure of \cite{Aquaro2013} in the context of linear panel data regression models with fixed effects includes to use of two different data transformations and applying the efficient weighted LS estimator of \cite{Gervini2002} and the reweighted LTS estimator of \cite{Cizek2010} on the transformed data. Moreover, for fixed effects panel data models, the robust estimators proposed by \cite{Bakar2015} have been developed by incorporating MM-centering procedure and the within-group generalized M-estimator (WGM) of \cite{Bramati2007}. For the estimation of both fixed effects and random effects regression models, \cite{Visek2015} has proposed a least weighted squares method based on mean-centering data while a weighted least squares technique using MM-estimate of location has been developed by \cite{Midi2018}. More recently, \cite{Beyaztas2020} have proposed robust versions of the OLS-based estimation procedures by using weighted likelihood estimating equations methodology within the panel data regression models framework. 

	In this study, we estimate model parameters using the density power divergence (DPD) based M-estimator proposed by \cite{MR1665873}. The main advantage of the proposed approach is that it balances the desired efficiency and the robustness of the estimators by controlling a tuning parameter. We also propose an adaptive method of choosing the tuning parameter, where there is no prior knowledge of outliers. We consider a linear panel data regression model with random effects as there are very few robust methods for the random-effects model. It is worth mentioning that, the estimation procedure for the fixed-effects model is relatively simple, and the theoretical properties are straightforward from the current work.
	
	The rest of the paper is organized as follows. We introduce the linear panel data model in Section \ref{sec:panel_data_model}. In Section \ref{sec:dpd}, we describe the density power divergence measure and the corresponding estimator for the linear panel data model. The theoretical properties, including the asymptotic distribution and the influence function of the MDPDE, are presented in Section \ref{sec:asymp}. We also propose a method to select the optimum DPD parameter by minimizing the asymptotic mean square error. Section \ref{Sec:Numerical Results} illustrates an extensive simulation study based on the proposed method and compares the results with the traditional techniques. The numerical results are further supported through a real data example from weather stations in Oman. Some concluding remarks are given in Section \ref{sec:conclusions} and the proofs of the technical results are shown in the Appendix.    
	
	\section{Linear Panel Data Models} \label{sec:panel_data_model}
	Let us consider the linear panel data regression model with a random sample as follows
	\begin{equation}
		y_{it} = x_{it}^T \beta + \alpha_i + \varepsilon_{it}, \ \  i = 1, 2, \cdots, N; \ \ t = 1, 2, \cdots, T,
		\label{panel_data_model}
	\end{equation}
	where the subscript $i$ represents an individual observed at time $t$. Here $\alpha_i$'s are the unobserved individual-specific effects (time-invariant characteristics), $\beta$ is a $K \times 1$ vector of regression coefficients and an element of the parameter space $\Theta$, $y_{it}$ and $x_{it}$'s are the response variable and the $K$-dimensional vector of explanatory variables, respectively and $\varepsilon_{it}$'s are the independent and identically distributed (iid) error terms with $E \left( \varepsilon_{it} \vert x_{i1}, \ldots, x_{iT}, \alpha_i \right) = 0$, $E \left( \varepsilon_{it}^2 \vert x_{i1}, \ldots, x_{iT}, \alpha_i \right) = \sigma_{\varepsilon}^2$ and $E \left( \varepsilon_{it} \varepsilon_{is} \vert x_{i1}, \ldots, x_{iT}, \alpha_i \right) = 0$ for $t \neq s$. The above panel data regression model can be represented in matrix form as 
	\begin{equation} 
		y = \alpha \otimes e_T + \X \beta + \varepsilon,
	\end{equation}
	where $y = \left( y_{1}, \cdots, y_{N} \right)^T$ is an $NT \times 1$ vector obtained by stacking observations $y_i = \left( y_{i1}, \cdots, y_{iT} \right)^T$ for individual $i=1,\cdots,N$. The $NT \times K$ matrix $X = \left( x_1, \cdots, x_N \right)^T$ is formed with regressors  $x_i = \left( x_{i1}^T, \cdots, x_{iT}^T \right)^T$, $\alpha$ is an $N \times 1$ vector consisting of the individual effects $\alpha_i$, $e_T$ is a $T \times 1$ vector of ones and $\otimes$ denotes the kronecker product. 
	
	If $\alpha_i$ is assumed to be random, then the random effects model can be succinctly written as

	\begin{equation} \label{Eq:3}
		y_{it} = x_{it}^T \beta + \alpha_i + \varepsilon_{it} = x_{it}^T \beta + \nu_{it},~~\alpha_i \sim \mathrm{iid}(0, \sigma_{\alpha}^2),~~\varepsilon_{it} \sim \mathrm{iid}(0, \sigma_{\varepsilon}^2),
	\end{equation}
	where $\nu_{it} = \alpha_i + \varepsilon_{it}$ denotes a compound error term with $\sigma_{\nu}^2 = \sigma_{\alpha}^2 + \sigma_{\varepsilon}^2$ and ${\rm cov}\left( \nu_{it},\nu_{is}\right) = \sigma_{\alpha}^2$ for $t \neq s$. $\alpha_i$'s are assumed to be uncorrelated with $\varepsilon_{it}$ and $x_{it}$. We further assume that $\alpha_i$ and $\epsilon_{it}$ are normally distributed for all $i$ and $t$.

	\section{Density Power Divergence} \label{sec:dpd}
	Let us consider a family of models $\{F_\theta, \theta \in \Theta_0\}$ with  density $f_\theta$. We denote $\mathcal{G}$ as the class of all distributions having densities with respect to Lebesgue measure. Suppose $G \in \mathcal{G}$ is the true distribution with density $g$. Then, the density power divergence (DPD) measure between the model density $f_\theta$ and the true density $g$ is defined as
	\begin{equation} 
		d_\gamma(f_\theta, g) = 
		\left\{
		\begin{array}{ll}
			\displaystyle{\int_y\left\{ f^{1+\gamma}_\theta(y)-\left( 1+\frac{1}{\gamma}\right) f^{\gamma }_\theta(y)g(y)+%
				\frac{1}{\gamma}g^{1+\gamma}(y)\right\} dy}, & \text{for}\mathrm{~}\gamma>0, \\%
			[2ex]
			\displaystyle{\int_y g(y)\log\left( \displaystyle\frac{g(y)}{f_\theta(y)}\right) dy,} & \text{for}%
			\mathrm{~}\gamma=0,%
		\end{array}
		\right. 
		\label{dpd}
	\end{equation}
	where $\gamma$ is a tuning parameter \citep{MR1665873}. Note that, $G$ is not necessarily a member of the model family $F_\theta$. Further, for $\gamma=0$, the DPD measure is obtained as a limiting case of $\gamma \rightarrow 0^+$, and  is same as the  Kullback-Leibler (KL) divergence. 
	Generally, given a parametric model, we estimate $\theta$ by minimizing the DPD measure with respect to $\theta$ over its parametric space $\Theta_0$. We call the estimator  the {\it minimum power divergence estimator} (MDPDE). 
	It is well-known that, for $\gamma=0$, minimization of the KL-divergent is equivalent to maximization of the log-likelihood function. Thus, the MLE can be considered as a special case of the MDPDE when $\gamma=0$. 
	
	Let $\theta = (\beta^T, \sigma_\alpha^2, \sigma_\epsilon^2)^T$ denote the parameter of the random effects model given in Equation (\ref{Eq:3}). We define the conditional probability density for disturbance terms, $\varepsilon_i + \alpha_i e_T  = y_i - x_i \beta $ as 
	\begin{equation}
		f_\theta(y_i | x_i) = (2\pi)^{-\frac{T}{2}} |\bO |^{-\frac{1}{2}} \exp \left\{ -\frac{1}{2} (y_i - x_i \beta)^T \bO^{-1} (y_i - x_i \beta)
		\right\},
	\end{equation}
	where  
	\begin{equation}
		\bO = E \left( \nu_i \nu_i^T \right) = \sigma_{\varepsilon}^2 \bI_T + \sigma_{\alpha}^2 e_T e_T^T 
		\label{omega}
	\end{equation}
	and  $\bI_T$ is the identity matrix of dimension $T$.
	In this case, we introduce the DPD measure based on the conditional density $f_\theta(y|x)$   as
	\begin{equation} 
		d_\gamma(f_\theta, g) = 
		\left\{
		\begin{array}{ll}
			\displaystyle{\int_x	\int_y\left\{ f^{1+\gamma}_\theta(y|x)-\left( 1+\frac{1}{\gamma}\right) f^{\gamma }_\theta(y|x)g(y|x)+%
				\frac{1}{\gamma}g^{1+\gamma}(y|x)\right\} h(x) dx dy,} & \text{for}\mathrm{~}\gamma>0, \\%
			[2ex]
			\displaystyle{\int_x \int_y g(y|x)\log\left( \displaystyle\frac{g(y|x)}{f_\theta(y|x)}\right) h(x) dx dy,} & \text{for}%
			\mathrm{~}\gamma=0,%
		\end{array}
		\right. 
		\label{dpd1}
	\end{equation}
	where $h(x)$ is the marginal probability density function of $X$ and $g(y|x)$ is the true conditional density of $Y$ given $X$.

	For $\gamma>0$, the DPD measure can empirically be written as
	\begin{equation} 
		\widehat{d}_\gamma(f_\theta, g) = 
		\frac{1}{N} \sum_{i=1}^{N}	\int_y  f^{1+\gamma}_\theta(y|x_i) dy - \frac{1+\gamma}{N \gamma} \sum_{i=1}^{N} f_\theta^{\gamma}(y_i | x_i) + c(\gamma),
		\label{dpd_emp}
	\end{equation}
	where $c(\gamma) = \frac{1}{N\gamma} \sum_{i=1}^N \int_y g_i^{1+\gamma}(y|x_i) dy$ does not depend on $\theta$.   Equation \eqref{dpd_emp} simplifies to
	\begin{equation}\label{dpd_est}
		\widehat{d}_\gamma(f_\theta, g) 
		= (2\pi)^{-\frac{T \gamma}{2}}  |\bO |^{-\frac{\gamma}{2}} (1+\gamma)^{-\frac{1}{2}} \left[ 1 - \frac{(1+\gamma)^{3/2} }{N\gamma}   \sum_{i=1}^{N}\exp\left[-\frac{\gamma}{2} B_i\right] \right]  + c(\gamma),
	\end{equation}	
	where $B_i = (y_i - x_i \beta)^T \bO^{-1} (y_i - x_i \beta)$. Using the Sherman--Morrison formula, we get
	\begin{equation}
		\bO^{-1} =  \frac{1}{\sigma_\epsilon^2} \bI_T - \frac{\sigma_\alpha^2 e_T e_T^T} {\sigma_\epsilon^2 ( \sigma_\epsilon^2 + T \sigma_\alpha^2 )},\ \ \
		|\bO|  = \sigma_\epsilon^{2(T-1)} ( \sigma_\epsilon^2 + T \sigma_\alpha^2 ).
		\label{inverse}
	\end{equation}
	It further simplifies $B_i$ as follows
	\begin{equation}
		B_i =    \frac{1}{\sigma_\epsilon^2}\sum_{t=1}^{T} (y_{it} - x_{it} \beta)^2 - \frac{\sigma_\alpha^2 } {\sigma_\epsilon^2 ( \sigma_\epsilon^2 + T \sigma_\alpha^2 )} \left\{ \sum_{t=1}^{T} (y_{it} - x_{it} \beta) \right\}^2 .
	\end{equation}
	
	The MDPDE of $\theta$ is then obtained by minimizing $	\widehat{d}_\gamma(f_\theta, g)$ over $\theta \in \Theta_0$. Note that, if the $i$-th observation is an outlier, then the value of $f_\theta(y_i | x_i)$ is very small compared to other samples. In that case, the second term of Equation \eqref{dpd_est} is negligible when  $\gamma>0$, thus the corresponding MDPDE becomes robust against outlier. On the other hand, when $ \gamma=0$, the KL divergent can be written as $\widehat{d}_\gamma(f_\theta, g) = -\sum_{i=1}^{N}\log f_\theta(y_i | x_i) +d$, where $d$ is independent of $\theta$. For an outlying observation, the KL divergenence measure diverges as $f_\theta(y_i | x_i) \rightarrow 0$. Therefore, the MLE breaks down in the presence of outliers as they dominate the loss function. In fact, the tuning parameter $\gamma$ controls the trade-off between  efficiency and robustness of the MDPDE -- robustness measure increases if $\gamma$ increases, but at the same time efficiency  decreases.

	Let $ \bar{x}_i = \frac{1}{T}\sum_{t=1}^T x_{it}$. The estimating equations of $\theta$ is then obtained from equation $\frac{\partial}{\partial \theta} \widehat{d}_\gamma(f_\theta, g) =0$ and they can be simplified as
	\begin{equation}
		\begin{split}
			&\sum_{i=1}^{N} \sum_{t=1}^{T}  x_{it} (y_{it} - x_{it} \beta)  \exp\left[-\frac{\gamma}{2} B_i\right]  = \frac{T \sigma_\alpha^2 } {( \sigma_\epsilon^2 + T \sigma_\alpha^2 )}  \sum_{i=1}^{N} \sum_{t=1}^{T} \bar{x}_i (y_{it} - x_{it} \beta) \exp\left[-\frac{\gamma}{2} B_i\right],\\
			& \gamma T |\bO |^{-1} \sigma_\epsilon^{2(T-1)}  \left\{  (1+\gamma)^{-\frac{1}{2}} - \frac{1+\gamma}{N \gamma}   \sum_{i=1}^{N}\exp\left[-\frac{\gamma}{2} B_i\right] \right\} \\
			& \hspace{1in} = -  \frac{(1+\gamma)}{N ( \sigma_\epsilon^2 + T \sigma_\alpha^2 )^2}   \sum_{i=1}^{N}\exp\left[-\frac{\gamma}{2} B_i\right]  \left\{ \sum_{t=1}^{T} (y_{it} - x_{it} \beta) \right\}^2 ,\\
			& \gamma T |\bO |^{-1}  \sigma_\epsilon^{2(T-2)} \left\{  \sigma_\epsilon^2 +  (T-1) \sigma_\alpha^2 \right\}   \left\{  (1+\gamma)^{-\frac{1}{2}} - \frac{1+\gamma}{N \gamma}   \sum_{i=1}^{N}\exp\left[-\frac{\gamma}{2} B_i\right] \right\} \\
			& \ \ =  \frac{(1+\gamma) }{N}   \sum_{i=1}^{N}\exp\left[-\frac{\gamma}{2} B_i\right]  \left[- \frac{1}{\sigma_\epsilon^4}\sum_{t=1}^{T} (y_{it} - x_{it} \beta)^2 + \frac{\sigma_\alpha^2 (2\sigma_\epsilon^2 + T \sigma_\alpha^2 )} {\sigma_\epsilon^4 ( \sigma_\epsilon^2 + T \sigma_\alpha^2 )^2} \left\{ \sum_{t=1}^{T} (y_{it} - x_{it} \beta) \right\}^2 \right].
		\end{split}	
	\end{equation}
	The MDPDE of $\theta$ is obtained by solving the above system of equations. One may use an iterative algorithm for this purpose or directly minimize the DPD measure in Equation \eqref{dpd_emp} with respect to $\theta \in \Theta_0$.

	\section{Asymptotic Distribution of the MDPDE} \label{sec:asymp}
	In this section, we present the asymptotic distribution of the MDPDE, when the data generating distribution $G(y|x)$ is not necessarily in the model family. Let us define the score function as $u_\theta(y_i | x_i) =		\frac{\partial}{\partial \theta} \log f_\theta(y_i | x_i)$. We write 
	\begin{equation}
		u_\theta(y_i | x_i) =	(u^T_\beta(y_i | x_i), u^T_{\sigma_\alpha^2}(y_i | x_i),  u^T_{\sigma_\epsilon^2}(y_i | x_i))^T. \label{u_score}
	\end{equation}
	Standard calculations show  that
	\begin{equation}
		\begin{split}
			u_\beta(y_i | x_i) &=		\frac{\partial}{\partial \beta} \log f_\theta(y_i | x_i) = 	 \frac{1}{\sigma_\epsilon^2}\sum_{t=1}^{T} x_{it} (y_{it} - x_{it} \beta) - \frac{ T \bar{x}_i \sigma_\alpha^2 } {\sigma_\epsilon^2 ( \sigma_\epsilon^2 + T \sigma_\alpha^2 )}  \sum_{t=1}^{T} (y_{it} - x_{it} \beta),\\
			u_{\sigma_\alpha^2}(y_i | x_i) &=		\frac{\partial}{\partial \sigma_\alpha^2} \log f_\theta(y_i | x_i)
			= 	 -\frac{T}{2  ( \sigma_\epsilon^2 + T \sigma_\alpha^2 )} 
			+ \frac{ 1} {2( \sigma_\epsilon^2 + T \sigma_\alpha^2 )^2} \left\{ \sum_{t=1}^{T} (y_{it} - x_{it} \beta) \right\}^2,\\
			u_{\sigma_\epsilon^2}(y_i | x_i) &=		\frac{\partial}{\partial \sigma_\epsilon^2} \log f_\theta(y_i | x_i) 
			= 	 -\frac{1}{2  ( \sigma_\epsilon^2 + T \sigma_\alpha^2 )} \left[ (T-1) \sigma_\epsilon^{2} ( \sigma_\epsilon^2 + T \sigma_\alpha^2 ) +  1\right]	\\
			& \hspace{1in} + \frac{1}{2\sigma_\epsilon^4}\sum_{t=1}^{T} (y_{it} - x_{it} \beta)^2 - \frac{\sigma_\alpha^2 (2\sigma_\epsilon^2 + T \sigma_\alpha^2 )} {2\sigma_\epsilon^4 ( \sigma_\epsilon^2 + T \sigma_\alpha^2 )^2} \left\{ \sum_{t=1}^{T} (y_{it} - x_{it} \beta) \right\}^2.
			\label{score_functions}
		\end{split} 
	\end{equation}
	For $i=1,2, \cdots, N$, we define
	\begin{equation}
		\begin{split}
			&\bJ^{(i)} = \int_y 	u_\theta(y | x_i) u^T_\theta(y | x_i) f_\theta^{1+\gamma}(y | x_i) dy + \int_y \Big\{ \bI_\theta(y | x_i) - \gamma 	u_\theta(y | x_i) u^T_\theta(y | x_i)\Big\} \Big\{ 	g(y | x_i) - f_\theta(y | x_i) \Big\}f_\theta^{\gamma}(y | x_i) dy,\\
			&\bK^{(i)} = \int_y 	u_\theta(y | x_i) u^T_\theta(y | x_i) f_\theta^{2\gamma}(y | x_i) g(y | x_i) dy - \xi^{(i)}  \xi^{(i)T} ,\\
			&\bI_\theta(y | x_i) = -\frac{\partial}{\partial \theta} u_\theta(y | x_i), \hspace{.4in} \xi^{(i)}  = \int_y u_\theta(y | x_i) f_\theta^{\gamma}(y | x_i) g(y | x_i) dy.
		\end{split}
	\end{equation}
	We further define $\bJ =  \lim_{N \rightarrow \infty} \frac{1}{N} \sum_{i=1}^{N} \bJ^{(i)}$, $\bK =  \lim_{N \rightarrow \infty} \frac{1}{N} \sum_{i=1}^{N} \bK^{(i)}$. For the asymptotic distribution of the MDPDE, we need the following assumptions:
	\begin{enumerate}
		\item[(A1)] The true density $g(y|x)$ is supported over the entire real line $\mathbb{R}$.
		
		\item[(A2)] There is an open subset $\omega \in \Theta_0$ containing the best fitting parameter $\theta$ such tat $\bJ$ is positive definite for all $\theta \in \omega$.
		
		\item[(A3)] There exist functions $M_{jkl}(x, y)$ such tat $|\partial^3 \exp[(y - x \beta)^T \bO^{-1} (y - x \beta)] /\partial \theta_j \partial \theta_k \partial \theta_l | \leq M_{jkl}(x, y)$ for all $\theta \in \omega$, where $\int_x \int_y |M_{jkl}(x, y)| g(y|x) h(x) dy dx < \infty$ for all $j, k$ and $l$. 
	\end{enumerate}
	\begin{theorem}
		\label{theorem:asymp}
		Under the regularity conditions (A1)--(A3), with probability tending to 1 as $n \rightarrow \infty$, there exists $\widehat{\theta}$, such that 
		\begin{enumerate}
			\item[(i)] $\widehat{\theta}$ is consistent for $\theta$, and
			\item[(ii)] the asymptotic distribution of $\widehat{\theta}$ is given by 
			\begin{equation}
				\sqrt{N}(\widehat{\theta} - \theta) \sim N_{K+2}(0, \bJ^{-1} \bK \bJ^{-1}).
			\end{equation}
		\end{enumerate}
	\end{theorem}
	\begin{proof} 
		The proof of the theorem is given in Appendix. 
	\end{proof}
	
	Note that, if the true distribution $g(y|x)$ is a member of the model family $f_\theta(y|x)$ for some $\theta \in \Theta_0$, then
	\begin{equation}
		\begin{split}
			&\bJ^{(i)} = \int_y 	u_\theta(y | x_i) u^T_\theta(y | x_i) f_\theta^{1+\gamma}(y | x_i) dy, \label{J_i}\\
			&\bK^{(i)} = \int_y 	u_\theta(y | x_i) u^T_\theta(y | x_i) f_\theta^{2\gamma+1}(y | x_i) dy - \xi^{(i)}  \xi^{(i)T} , \\
			& \xi^{(i)}  = \int_y u_\theta(y | x_i) f_\theta^{\gamma+1}(y | x_i) dy. 
		\end{split}
	\end{equation}
	
	\noindent
	In this case, the symmetric matrix $\bJ^{(i)}$ can be partitioned as
	\begin{equation}
		\bJ^{(i)} = 
		\begin{bmatrix}
			\bJ_\beta^{(i)}  & 	\bJ_{\beta, \ \sigma_\alpha^2}^{(i)} &  	\bJ_{\beta, \ \sigma_\epsilon^2}^{(i)}\\
			
			. & \bJ_{\sigma_\alpha^2}^{(i)} & \bJ_{\sigma_\alpha^2, \ \sigma_\epsilon^2}^{(i)}\\
			. & . & \bJ_{\sigma_\epsilon^2}^{(i)} 
		\end{bmatrix}, \label{J_theta}
	\end{equation}
	and in Appendix \ref{appendix:J}, it is shown that
	\begin{equation}
		\begin{split}
			\bJ_\beta^{(i)} & = M \sigma_\epsilon^{-4} \Bigg[ \sigma_\epsilon^2 \sum_{t=1}^{T} x_{it} x_{it}^T   +  T^2 \sigma_\alpha^2 \left( \frac{ T \sigma_\alpha^2 }{( \sigma_\epsilon^2 + T \sigma_\alpha^2 )} -1 \right) \bar{x}_i \bar{x}_i^T  \Bigg] ,\\
			\bJ_{\sigma_\alpha^2}^{(i)} & = \frac{M T^2 ( \gamma^2 + 2 )} {4 (1+\gamma) ( \sigma_\epsilon^2 + T \sigma_\alpha^2 )^2} ,\\
			\bJ_{\sigma_\epsilon^2}^{(i)}  & =  \frac{  M T^2 (\gamma - 1) \left[  \sigma_\epsilon^2 + (T-1) \sigma_\alpha^2 \right]^2} {4\sigma_\epsilon^4  ( \sigma_\epsilon^2 + T \sigma_\alpha^2 )^2}  \\
			& \hspace{1cm} + \frac{MT}{4\sigma_\epsilon^8}   \left[(T+2)  \sigma_\epsilon^4  +  2(T+2)  \sigma_\epsilon^2\sigma_\alpha^2 + 3T \sigma_\alpha^4 \right]    + \frac{3MT^2\sigma_\alpha^4 (2\sigma_\epsilon^2 + T \sigma_\alpha^2 )^2} {4\sigma_\epsilon^8 (1+\gamma) ( \sigma_\epsilon^2 + T \sigma_\alpha^2 )^2}  \\
			& \hspace{1cm}  - \frac{ T M (1+\gamma)\sigma_\alpha^2 (2\sigma_\epsilon^2 + T \sigma_\alpha^2 )} {2\sigma_\epsilon^8 ( \sigma_\epsilon^2 + T \sigma_\alpha^2 )^2} \Big[ (T+2)  \sigma_\epsilon^4 + (T^2 + 2T +3) \sigma_\epsilon^2 \sigma_\alpha^2 + 3(T^2 -T + 1 ) \sigma_\alpha^4 \Big] ,\\
			\bJ_{\beta, \ \sigma_\alpha^2}^{(i)} & = 0,\ \
			\bJ_{\beta, \ \sigma_\epsilon^2}^{(i)}  = 0,\\
			\bJ_{\sigma_\alpha^2, \ \sigma_\epsilon^2}^{(i)} & = \frac{T M  (1+\gamma) }{4 \sigma_\epsilon^4 ( \sigma_\epsilon^2 + T \sigma_\alpha^2 )^2}  \Big[  2(T+1)  \sigma_\epsilon^4 + (2T^2  + T +3) \sigma_\epsilon^2 \sigma_\alpha^2   + 3(T^2 -T +1) \sigma_\alpha^4 \Big] \\
			& \hspace{1cm} - \frac{3MT^2\sigma_\alpha^2 (2\sigma_\epsilon^2 + T \sigma_\alpha^2 )} {4\sigma_\epsilon^4 (1+\gamma) ( \sigma_\epsilon^2 + T \sigma_\alpha^2 )^2}  	 - \frac{M T^2  \left[  (T-1) \sigma_\alpha^2 +  \sigma_\epsilon^2 \right]} {2\sigma_\epsilon^2 ( \sigma_\epsilon^2 + T \sigma_\alpha^2 )^2} ,
		\end{split}
	\end{equation}
	where 
	\begin{equation}
		M = (2\pi)^{-\frac{T \gamma}{2}}  (1+\gamma)^{-\frac{T+2}{2}}  \sigma_\epsilon^{-\gamma(T-1)} ( \sigma_\epsilon^2 + T \sigma_\alpha^2 )^{-\frac{\gamma}{2}}  . \label{M1}
	\end{equation} 
	
	\noindent
	Similarly, $\xi^{(i)}$ can be partitioned as
	$\xi^{(i)} = \left(\xi_\beta^{(i)T}, \xi_{\sigma_\alpha^2}^{(i)} , \xi_{\sigma_\epsilon^2}^{(i)} \right)^T$,
	and in Appendix \ref{appendix:xi}, it is shown that
	\begin{equation}
		\xi_\beta^{(i)} = 0,\ \ \
		\xi_{\sigma_\alpha^2}^{(i)}  = 	 -\frac{MT \gamma}{2  ( \sigma_\epsilon^2 + T \sigma_\alpha^2 )},\mbox{ and }
		\xi_{\sigma_\epsilon^2}^{(i)}  =  -\frac{M T \gamma \left[  \sigma_\epsilon^2 + (T-1) \sigma_\alpha^2 \right]}{2 \sigma_\epsilon^2 ( \sigma_\epsilon^2 + T \sigma_\alpha^2 ) }.
		\label{xi_i_model}
	\end{equation}
	Note that if we write the matrix $\bJ^{(i)}$ as a function of $\gamma$, i.e. $\bJ^{(i)} \equiv \bJ^{(i)}(\gamma)$, then we have
	\begin{equation}
		\bK^{(i)} = \bJ^{(i)}(2\gamma)   - \xi^{(i)}  \xi^{(i)T}.
	\end{equation}
	Moreover, $\xi_\beta^{(i)}$ is constant for all values of $i=1, 2, \cdots, N$. Therefore, $\bK$ can be written as
	\begin{equation}
		\bK =  \lim_{N \rightarrow \infty} \frac{1}{N} \sum_{i=1}^N \bJ^{(i)}(2\gamma)   -  \xi^{(i)}  \xi^{(i)T}.
	\end{equation}
	
	Further calculations show that the variance of each component of $\widehat{\beta}$ increases as $\gamma$ increases. Therefore, the efficiency of the MDPDE decreases as $\gamma$ increases -- the MLE being the most efficient estimator. However, our simulation studies show that the loss of efficiency is not severe.

	\subsection{Influence Function of the MDPDE}
	We further access the extent of the resistance to outliers of our proposed estimator using the influence function approach of \cite{MR829458}. It measures the rate of asymptotic bias of an estimator to infinitesimal contamination in the distribution. A bounded influence function suggests that the corresponding estimator is robust against extreme outliers. Note that, the MDPDE is an M-estimator \citep{MR606374} as the estimating equation can be written as $\sum_i \Psi_\theta(y_i | x_i) = 0$, where
	\begin{equation}
		\Psi_\theta(y_i | x_i) =  u_\theta(y_i | x_i) f_\theta^\gamma(y_i | x_i) - \int_y u_\theta(y | x_i) f_\theta^{1+\gamma}(y | x_i) dy.
	\end{equation}
	This is obtained by differentiating $\widehat{d}_\gamma(f_\theta, g)$ with respect to $\theta$ from Equation \eqref{dpd_emp}. Let $G(y|x)$ be the true conditional distribution function $Y$ given $X$  and $\theta = T_\gamma(G)$ be the functional for the MDPDE.  Following \cite{MR1665873}, the influence function of the MDPDE is given by
	\begin{equation}
		IF((x, y), T_\gamma, G) = \bJ^{-1} \left\{ u_\theta(y | x) f_\theta^\gamma(y | x) - \xi^{(i)} \right\},
	\end{equation}
	where  $\bJ$ is evaluated at the model when $g=f_\theta$, and $\xi^{(i)}$, given in Equation \eqref{xi_i_model}, is a fixed vector that  does not depend on index $i$.
	\begin{remark}
		Note that the score function $u_\theta(y | x)$ in Equation \eqref{u_score} is unbounded in both $x$ and $y$. As a result, the influence function of the MLE, i.e., the MDPDE with $\gamma=0$, is unbounded  for the panel data regression model.  On the other hand,  $u_\theta(y | x) f_\theta^\gamma(y | x)$ is bounded in $y$ when $\gamma>0$ as the corresponding terms can be written as $y\exp(y^2)$. So, the influence function of the MDPDE of $\theta$ is bounded in $y$ when $\gamma>0$. Moreover, $IF((x, y), T_\gamma, G)$ tends to zero as $|y| \rightarrow \infty$, indicating a redescending effect for large vertical outliers. The higher the value of $\gamma$, the larger the down-weighting effect to the outliers. However, this influence function could still be unbounded in $x$, but only for small values of $|y|$ and simultaneously large values of $||x||$. This implies that good leverage points have the strongest effect on the MDPDE. But, if the leverage points are also vertical outliers, then the MDPDE will not be sensitive to those observations.
	\end{remark}
	
	\subsection{Choice of the Optimum $\gamma$} \label{sec:opt_dpd}
	One important use of the asymptotic distribution of the penalized MDPDE is in selecting the optimum value of the DPD parameter $\gamma$. 
	In practice,  $\gamma$ is chosen by the user depending on the desired level of robustness measure at the cost of efficiency. Alternatively, we may select a data-driven optimum $\gamma$. Following  \cite{Warjones}, we minimize  the  mean square error (MSE) of $\widehat{\beta}$ to obtain the optimum value of $\gamma$ adaptively. Suppose $\Sigma_\beta$ the asymptotic variance of $\beta$ obtained from Theorem \ref{theorem:asymp} assuming that the true distribution belongs to the model family. Let $\widehat{\bsigma}_\beta$ be the estimate of $\bsigma_\beta$. The  empirical estimate of the MSE, as the function of a pilot estimator $\beta^P$, is given by
	\begin{equation}
		\widehat{MSE}(\gamma) = (\widehat{\beta} - \beta^P)^T (\widehat{\beta} - \beta^P) +\tr(\widehat{\bsigma}_\beta).
		\label{adaptive_gamma}
	\end{equation}
	In particular, we recommend that a robust estimator, such as the MDPDE with $\gamma=0.5$, to be used as a pilot estimator. One may also iterate this process by taking the previous stage's optimum $\gamma$ as the pilot estimator of the current stage, and proceed until convergence. In our numerical examples, we have used this iterative procedure.

	\section{Numerical Results} \label{Sec:Numerical Results}

	To investigate the performances of our proposed method, we conduct an extensive simulation study under  different sample sizes and  different types of outliers. We compare the performance with the weighted likelihood estimator (WLE) of \cite{Beyaztas2020} as a robust alternative method and the generalized least squares (GLS) and ordinary least squares (OLS) estimators. As the robustness properties of the MDPDE depend on the choice of the tuning parameter, we have taken four fixed values of $\gamma = 0.1, 0.2, 0.3$ and $0.4$ along with the data-driven adaptive choice that minimizes the MSE as discussed in Section \ref{sec:opt_dpd}. The pilot estimator is used iteratively until convergence. 

	We consider the random effects model given in Equation (\ref{Eq:3}). The vector of regression coefficients is taken as $\beta  = \left(2, 2.4, -1.2, 1.6, -0.5 \right)^{T}$, where the first component is the intercept term. The individual-specific effects $\alpha_i$s are generated from the standard normal distribution.  The explanatory variables $x_{itk}$ for $k = 2, 3,4, 5$ are generated as follows:
	\begin{equation}
		x_{itk} \sim \left \{\begin{array}{ll}
			\chi_2^2 - 2 &~ \text{for}~k=2, \\
			N \left( 0, 1 \right) &~ \text{for}~k > 2,
		\end{array}
		\right.
	\end{equation}
	where $\chi_2^2$ represents the chi-square distribution with 2 degrees of freedom. One regressor is generated from a skewed distribution  to avoid a symmetric experimental design.

	To evaluate the performance of each estimator based on $S = 1000$ simulations, we compute the mean squared errors (MSEs) of $\widehat{\beta}$ given by
	\begin{equation}
		MSE = \frac{1}{S} \sum_{s=1}^S \Big\| \widehat{\beta}^s-\beta \Big\|^2,~~s=1, 2, \cdots, S,
	\end{equation}
	where $\widehat{\beta}^s$ is the estimate obtained from the $s$-th replication and $\beta$ is the true value of the parameter. As all estimators are root-$N$ consistent, the MSE values presented here are multiplied by $N$. Thus, in this section, we denote MSE as the MSE of $\sqrt{N}\widehat{\beta}$. In addition, the outliers deleted mean prediction error (MPE) is calculated to assess the predictive performance of the methods under consideration. We define MPE as
	\begin{equation}
		MPE = \frac{1}{S} \sum_{s=1}^S \left[\frac{1}{NT - m^{(s)}}\sum_{i=1}^N \sum_{t=1}^T \left(1-c_{it}^{(s)}\right) \left( y_{it}^{(s)} - \widehat{y}_{it}^{(s)} \right)^2 \right],
	\end{equation}
	where $y_{it}^{(s)}$ and $\widehat{y}_{it}^{(s)}$ are the actual and predicted value of the $(i,t)$-th element of $y$, respectively. The indicator variable $c_{it}^{(s)}$ is 1 if $y_{it}^{(s)}$ is an outlier, and 0 otherwise. Here  $m^{(s)}$  denotes the total number of outliers in the $s$-th simulated sample. In pure data, we assume that there is no outlier, i.e., $c_{it}^{(s)}=0$ for all $i, t$ and $s$. 

	\begin{table}[ht]
		\centering
		\caption{The MSEs of $\sqrt{N}\widehat{\beta}$ (first rows)  and the MPEs (second rows) of all estimators when $N = 25,  50, 100, 200$ for a fixed time dimension $T = 5$ (first 4 columns) and $T = 10, 15, 20, 25$ for a fixed cross-sectional dimension $N = 100$ (last 4 columns).}
		\begin{tabular}{lrrrr|rrrr}
			\toprule
			& N=25 & N=50 & N=100 & N=200 & T=10 & T=15 & T=20 & T=25 \\ 
			\hline
			DPD(0.1) & 2.1526 & 2.1166 & 2.1962 & 2.0683 & 1.5015 & 1.3676 & 1.3186 & 1.2810 \\ 
			& 1.8857 & 1.9227 & 1.9504 & 1.9616 & 1.9554 & 1.9638 & 1.9655 & 1.9626 \\ \hline
			DPD(0.2) & 2.2920 & 2.2538 & 2.3052 & 2.1860 & 1.6267 & 1.4734 & 1.3962 & 1.3304 \\ 
			& 1.8911 & 1.9252 & 1.9515 & 1.9622 & 1.9566 & 1.9647 & 1.9663 & 1.9632 \\ \hline
			DPD(0.3) & 2.4397 & 2.3918 & 2.4186 & 2.3035 & 1.6968 & 1.5043 & 1.4352 & 2.6966 \\ 
			& 1.8970 & 1.9278 & 1.9527 & 1.9628 & 1.9572 & 1.9650 & 1.9664 & 1.9678 \\ \hline
			DPD(0.4) & 2.5523 & 2.4959 & 2.5057 & 2.3917 & 1.7229 & 1.6193 & 2.4267 & 1.3455 \\ 
			& 1.9016 & 1.9299 & 1.9536 & 1.9633 & 1.9575 & 1.9651 & 1.9687 & 1.9809 \\ \hline
			DPD(Opt.) & 2.1096 & 2.0641 & 2.1633 & 2.0286 & 1.4342 & 1.2847 & 1.3621 & 1.1063 \\ 
			& 1.8843 & 1.9220 & 1.9501 & 1.9615 & 1.9549 & 1.9632 & 1.9741 & 1.9650 \\ 
			Mean Opt. $\gamma$ & 0.0134 & 0.0094 & 0.0057 & 0.0042 & 0.0067 & 0.0096 & 0.3211 & 0.1433 \\ \hline
			OLS & 2.0906 & 2.0456 & 2.1543 & 2.0114 & 1.4169 & 1.2629 & 1.2213 & 1.2062 \\ 
			& 1.8840 & 1.9221 & 1.9501 & 1.9615 & 1.9548 & 1.9631 & 1.9647 & 1.9617 \\ \hline
			GLS & 2.0975 & 2.0502 & 2.1555 & 2.0108 & 1.4170 & 1.2629 & 1.2212 & 1.2062 \\ 
			& 1.8829 & 1.9218 & 1.9500 & 1.9614 & 1.9548 & 1.9630 & 1.9647 & 1.9617 \\ \hline
			WLE & 2.1164 & 2.0619 & 2.1581 & 2.0217 & 1.4220 & 1.2662 & 1.2207 & 1.2050 \\ 
			& 1.8832 & 1.9217 & 1.9500 & 1.9614 & 1.9548 & 1.9630 & 1.9647 & 1.9617 \\ 
			\bottomrule
		\end{tabular}
		\label{table:pure_data}
	\end{table}

	\begin{table}[ht]
		\centering
		\caption{The MSEs of $\sqrt{N}\widehat{\beta}$ (first rows)  and the MPEs (second rows) of all estimators in the presence of $p = 5\%,  7.5\%, 10\%$ vertical outliers at random location (first 3 columns), concentrated panels (middle 3 columns) 
			and concentrated vertical outliers with 50\% leverage points (last 3 columns). In all cases, $N=100$ and $T=5$.}
		\small{
			\begin{tabular}{lrrr|rrr|rrr}
				\toprule
				& p=5\% & p=7.5\% & p=10\% & p=5\% & p=7.5\% & p=10\% & p=5\% & p=7.5\% & p=10\% \\ 
				\hline
				DPD(0.1) & 2.9669 & 13.8441 & 42.7548 & 2.4000 & 9.5943 & 31.3995 & 2.4125 & 11.1179 & 30.9229 \\ 
				& 1.9899 & 2.0879 & 2.3769 & 1.9868 & 2.0474 & 2.2705 & 1.9773 & 2.0731 & 2.2682 \\ \hline
				DPD(0.2) & 2.8797 & 3.1348 & 3.4681 & 2.2859 & 2.3152 & 2.4898 & 2.2980 & 2.3651 & 2.4573 \\ 
				& 1.9904 & 1.9837 & 1.9922 & 1.9862 & 1.9769 & 1.9752 & 1.9761 & 1.9867 & 1.9806 \\ \hline
				DPD(0.3) & 2.9856 & 3.2175 & 3.4768 & 2.3926 & 2.4257 & 2.6097 & 2.4130 & 2.4661 & 2.5472 \\ 
				& 1.9917 & 1.9848 & 1.9927 & 1.9874 & 1.9781 & 1.9763 & 1.9773 & 1.9879 & 1.9817 \\ \hline
				DPD(0.4) & 3.0478 & 3.2571 & 3.4962 & 2.4734 & 2.5083 & 2.6970 & 2.4990 & 2.5465 & 2.6240 \\ 
				& 1.9924 & 1.9853 & 1.9930 & 1.9883 & 1.9789 & 1.9772 & 1.9782 & 1.9889 & 1.9825 \\ \hline
				DPD(Opt.) & 2.7893 & 3.0803 & 3.4154 & 2.2247 & 2.2450 & 2.4219 & 2.2244 & 2.3169 & 2.4057 \\ 
				& 1.9891 & 1.9826 & 1.9911 & 1.9854 & 1.9763 & 1.9747 & 1.9753 & 1.9862 & 1.9802 \\ 
				Mean Opt. $\gamma$ & 0.1255 & 0.1413 & 0.1584 & 0.1268 & 0.1414 & 0.1497 & 0.1278 & 0.1458 & 0.1539 \\ \hline
				OLS & 30.7954 & 64.5063 & 108.1523 & 27.7604 & 66.3568 & 101.2097 & 28.6091 & 60.7410 & 84.8934 \\ 
				& 2.2745 & 2.6071 & 3.0426 & 2.2359 & 2.6179 & 2.9732 & 2.2499 & 2.5823 & 2.8199 \\ \hline
				GLS & 30.5074 & 64.1230 & 107.5897 & 27.6653 & 66.1988 & 101.0173 & 25.7178 & 60.2970 & 91.2663 \\ 
				& 2.2648 & 2.5944 & 3.0274 & 2.2379 & 2.6206 & 2.9763 & 2.2202 & 2.5727 & 2.8791 \\ \hline
				WLE & 4.5795 & 5.4534 & 6.0592 & 23.3767 & 60.8732 & 94.7575 & 20.3808 & 52.6727 & 83.3365 \\ 
				& 2.0010 & 1.9973 & 2.0103 & 2.1954 & 2.5670 & 2.9128 & 2.1653 & 2.4960 & 2.7991 \\
				\bottomrule
			\end{tabular}
		}
		\label{table:contaminated_data}
	\end{table}

	\subsection{Sample sizes} \label{Sec:Sample size}

	The performance of the estimators is examined under standard normal errors $\varepsilon_{it} \sim$ $\text{N}(0, 1)$ for the increasing number of cross-sectional units ($N = 25, 50, 100,  200$) by keeping time period fixed at $T = 5$ and for the increasing values of time periods ($T = 10, 15, 20, 25$) when the number of cross-sectional units is fixed at $N=100$. Table  \ref{table:pure_data} illustrates the simulation results. For fixed $N$ and $T$, the performance of all estimators are very similar. Theoretically, it is shown that the MDPDE loses efficiency as $\gamma$ increases. It is also observed in this table, although the difference is very small. The mean value of the data-dependent optimum DPD parameter $\gamma$ is also reported in this table. In most cases, they are close to zero, and therefore, its MSE and MPE are almost identical to the OLS estimator. For $T=20$ and 25, the mean optimum $\gamma$ parameters comes out to be large due to very small asymptotic variance. But in the end, they also produce similar MSE and MPE as the OLS and GLS.  
	The table shows that, when $T$ is fixed, the MSEs of $\sqrt{N}\widehat{\beta}$ for an estimator seems to be a constant for all values of $N$. It justifies that all estimators, including our proposed one, are root-$N$ consistent.

	\subsection{Outliers} \label{Sec:Outliers}

	In the following simulation studies, the robustness properties of the estimators are evaluated in the presence of different types of outliers. The cross-sectional size $N = 100$ and time periods $T = 5$ are chosen for the panel size consisting of a total of $500$ observations. Three levels of contamination ($p$) are considered as 5\%, 7.5\%, and 10\%. To generate contaminated data, outliers are inserted using two different ways -- random contamination and concentrated contamination as explained in \cite{Bramati2007}. The random contamination is obtained by distributing outlying data points randomly over all observations. On the other hand, the outliers clustered in all time points of some randomly selected blocks when generating concentrated contamination. The contamination schemes considered are explained below based on the types of outliers.

	\begin{enumerate}
		\item The random outliers in the $y$-direction, namely, random vertical outliers  are obtained by replacing the original standard normal errors to $\varepsilon_{it} \sim$ $\text{N}(10, 1)$ in panel data model given in Equation \eqref{Eq:3}.
		
		\item Concentrated vertical outliers are generated by substituting all errors in the selected random  blocks by $\varepsilon_{it} \sim$ $\text{N}(10, 1)$. 
		
		\item To obtain random leverage points, i.e., random contamination in both $y$-direction and $x$-direction, the same rule in the second scheme is applied to generate the error variable. Then, 50\% values of the regressors corresponding to the outlying points  are replaced by  $x_{itk} \sim\text{N}(5, 1)$. 
		
	\end{enumerate}

	These three sets of simulation results are presented in Table \ref{table:contaminated_data}. They clearly indicate that the proposed estimator MDPDE (except   $\gamma=0.1$) outperforms the conventional estimators (OLS and GLS) and the robust estimator WLE in all cases. 
	The OLS and GLS methods result in obtaining severely distorted estimates of the parameters in the presence of outliers and, as a result, produce the largest MSE and MPE. The WLE works well in the presence of random contamination. But it breaks down when outliers are clustered, or there are leverage points. For small values of $\gamma$, the MDPDEs are not robust against outliers. So DPD(0.1) has large MSE and MPE, however, they are smaller than the OLS estimator.  On the other hand, considering both performance metrics MSE and MPE, other DPD estimators are almost insensitive to different choices of contamination schemes and levels. Moreover, the MSE and MPE are somehow similar to the corresponding values in Table \ref{table:pure_data} for the pure data. It suggests that the performance of the MDPDE for large values of $\gamma$ are robust against outliers. The data-dependent optimum MDPDE automatically selects a high value of $\gamma$ where the estimated MSE is minimized. We notice that the mean value of the optimum $\gamma$ increases with the contamination proportion. 
	
	Other than the above contamination schemes, we have used different sets of $N$ and $T$, different error distributions (eg. chi-square and t-distribution), and placed the center of outliers at various points and observe a similar pattern in those cases. These outcomes suggest that the performance of the MDPDE is almost identical with the classical efficient methods, like OLS and GLS, in the pure data. On the other hand, in contaminated data, the MDPDEs with large values of $\gamma$ yield an accurate and precise estimate of the parameters even when the WLE fails. The data-dependent optimum MDPDE successfully produces the optimum performance and properly balances the efficiency in pure data and robustness properties in the contaminated data. As in real-data analysis, we generally do not have prior knowledge of the proportion and size of outliers, an adaptive choice of the DPD tuning parameter plays an important role. From this simulation study, there is also substantial indication and evidence of the theoretical robustness properties and root-$N$ consistency result of the MDPDE derived in this paper.  

	\subsection{Case Study} \label{Sec:Case Study}

	In this section, we apply the proposed methodology to analyze the Oman weather dataset, which is available on the National Center for Statistics \& Information at {\it{https://data.gov.om/bixytwb/weather}}. This dataset consists of a total of 660 observations ($N = 55$, $T = 12$) covering a cross-section of 55 stations across Oman over the period January 2018 to December 2018. The list of weather stations is given in Table \ref{tab:stations} of Appendix \ref{sec:weather_station}. Our interest is in the relationships between the monthly minimum evaporation (mm) and  two regressors -- minimum temperature  ($^{\circ}$C) and minimum humidity (\%). For this dataset, the following panel data regression model is conducted:
	\begin{equation}
		E_{it} = \beta_0 + \beta_1 Temp_{it} + \beta_2 H_{it} + \nu_{it}~~ i=1, 2, \cdots, 55;~t=1, 2, \cdots, 12,
	\end{equation}
	where $E$ denote the monthly minimum evaporation (mm) as a response variable, $Temp$ and $H$, respectively, represent the minimum temperature ($^{\circ}$C) and minimum humidity (\%) as explanatory variables. The scatter plots of monthly evaporation versus the temperature and humidity are given in Figure \ref{Fig:1}. The first plot shows a cluster on the top that may arise due to unusual points in either the response variable or the explanatory variables. There are also a few large values of humidity in the second plot; moreover, some observations are very much isolated from the main region. 

	In real data, as we do not have prior knowledge of outliers, trimmed mean prediction errors are used to evaluate the performance of different estimators. 
	The estimated individual coefficients, trimmed MPEs and the percentage of trimmed MPE increased over the MDPDE are reported in Table \ref{tab:environmental_data}. Three different values of trimming percentages are considered as $p = 20\%, 10\%$, and $5 \%$. The optimum value of the DPD parameter based on the iterative method, discussed in Section \ref{sec:opt_dpd}, comes out to be $\gamma=0.1439$ for this data set. From Table \ref{tab:environmental_data}, we observe that the  MDPDE with the optimum $\gamma$ parameter yields the smallest MPE for all trimming percentages. The higher the trimming proportions, the better the performance from the MDPDE. For example, in 20\% trimming case, the MDPDE gives 16.75\%, 25.41\% and 8.07\% better prediction than the OLS, GLS and WLE, respectively, for the remaining 80\% observations.  By construction, the OLS estimator produces the least MPE in the linear panel data regression model when all observations are considered. Thus, the corresponding regression line moves closer to the outliers by sacrificing the prediction power for good observations. On the other hand, the MDPDE gives a better fit for those observations that are close to the fitted model and are least affected by  outliers. Being a robust method, the WLE gives a relatively smaller MPE; however, the MDPDE shows further improvement.  All the results clearly demonstrate that the proposed estimator has considerably better predictive ability compared to the traditional and robust WLE methods.
	\begin{figure}
		\centering
		\includegraphics[width=9cm]{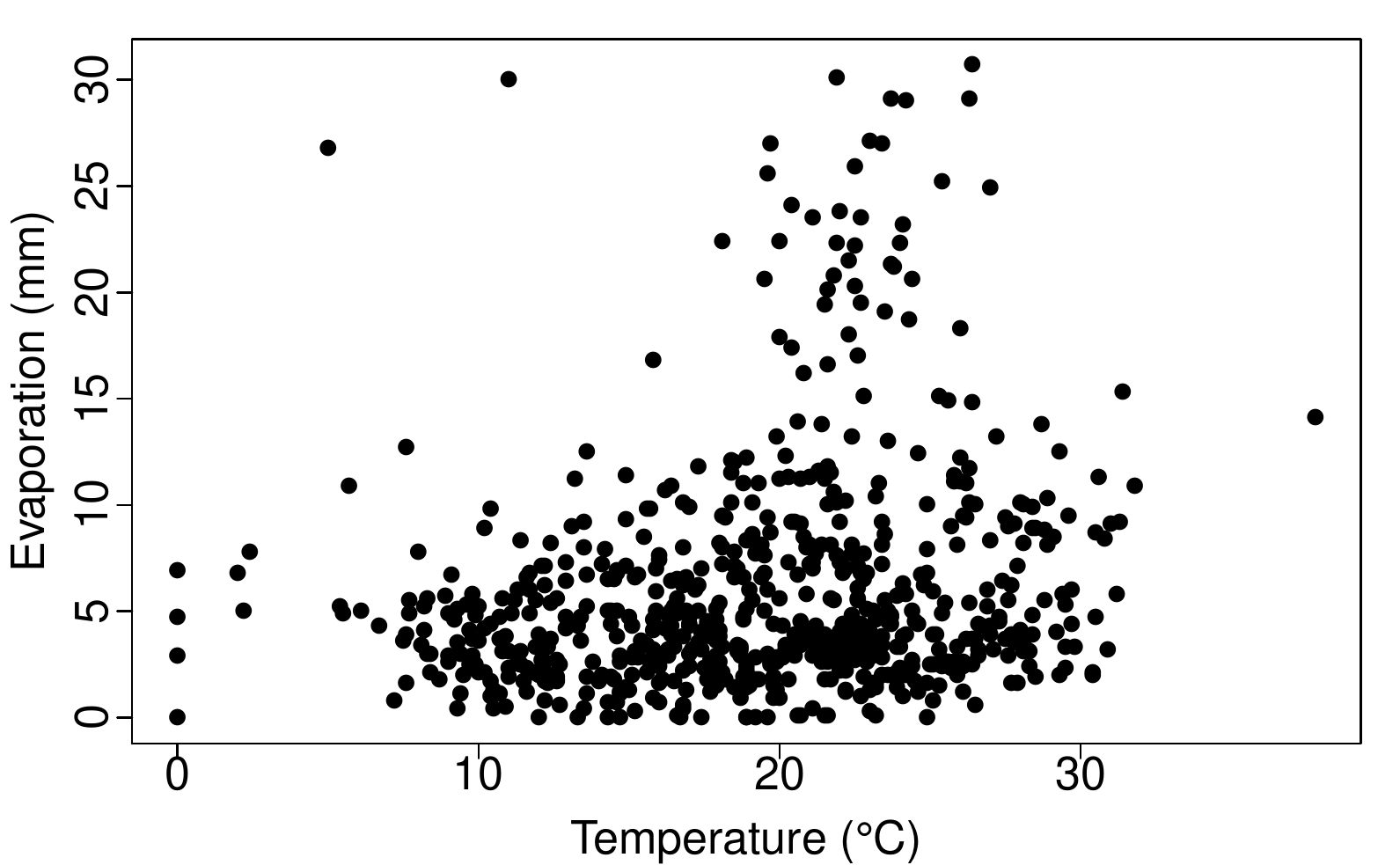}
		\includegraphics[width=9cm]{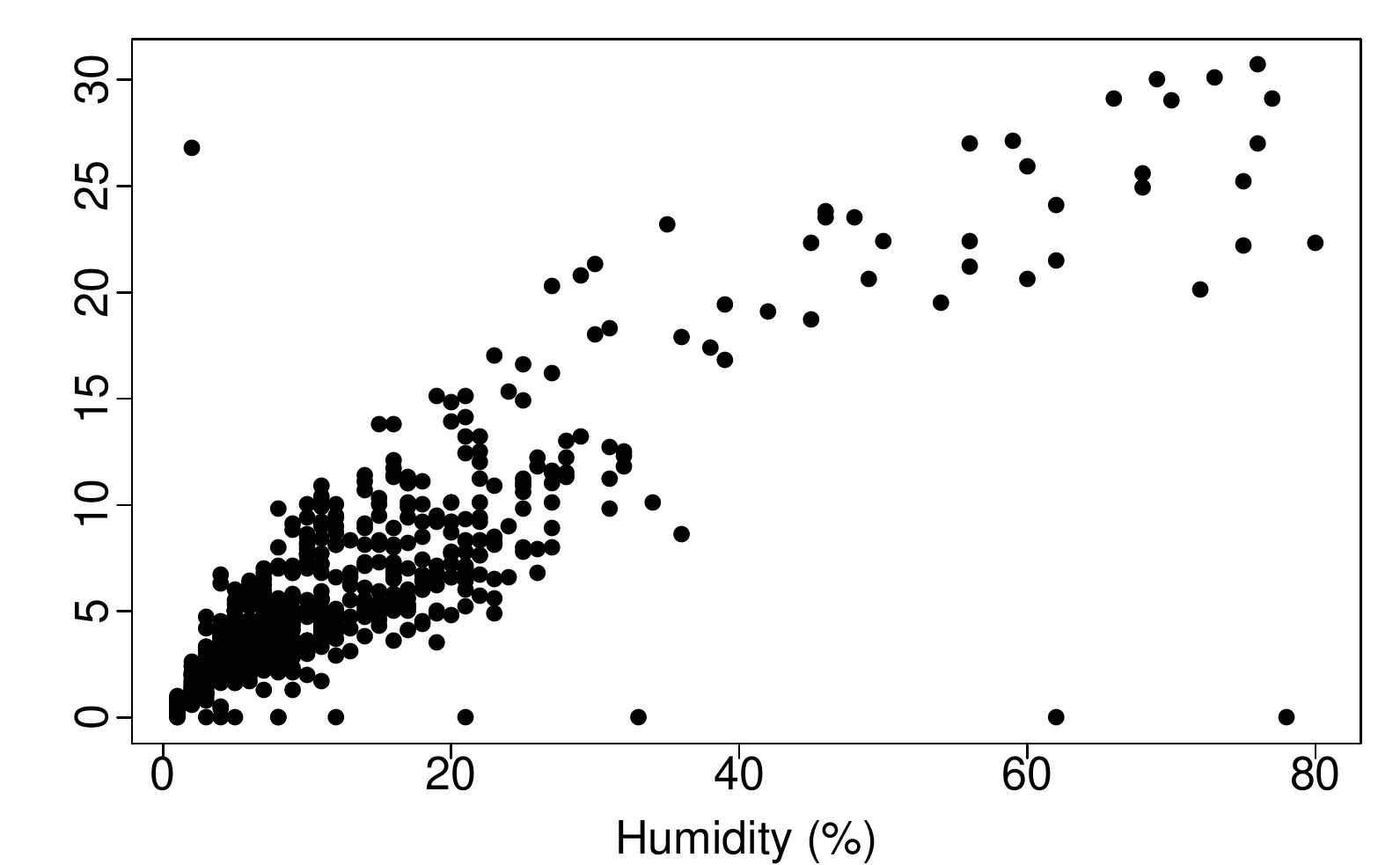}
		\caption{Scatter plots of the monthly temperature versus evaporation (left panel) and humidity versus evaporation (right panel) for Oman weather data.}
		\label{Fig:1}
	\end{figure}	

	\begin{table}[h!]
		\centering
		\caption{The estimates of individual coefficients and MPEs of all estimators for different trimmings.}
		\begin{tabular}{ccrrrr}
			\toprule
			Trimming & Estimates & DPD(Opt.) & OLS & GLS & WLE\\
			\hline
			& $\widehat{\beta}_0$ & $-2.4570$ & $-1.6281$ & $-1.7009$ & $-2.6621$\\
			& $\widehat{\beta}_1$ & 0.1764 & 0.1732 & 0.1860 & 0.2030\\
			Full Data & $\widehat{\beta}_2$ & 0.4039 & 0.3429 & 0.3293 & 0.3868\\
			& $\widehat{\sigma}_\alpha$ & 0.4674 & 1.2530 & 1.1945 & 0.0567\\
			& $\widehat{\sigma}_\epsilon$ & 1.5860 & 5.5990 & 6.0516 & 2.1475 \\
			\hline
			$p=20\%$ & MPE & 0.9351 & 1.0918 & 1.1727 & 1.0107\\
			& Increase & -- & 16.75\% & 25.41\% & 8.07\%\\
			\hline
			$p=10\%$ & MPE & 1.5704 & 1.6937 & 1.8000 & 1.6103\\
			& Increase & -- & 7.85\% & 14.62\% & 2.54\%\\
			\hline
			$p=5\%$ & MPE & 2.1579 & 2.2965 & 2.4379 & 2.1643\\
			& Increase & -- & 6.42\% & 12.97\% & 0.30\%\\
			\bottomrule
		\end{tabular}
		\label{tab:environmental_data}
	\end{table}

	\section{Conclusions} \label{sec:conclusions} 
	We have proposed a robust procedure for estimating the parameters of the linear panel data regression model with random effects using the density power divergence. The efficiency and robustness properties of the proposed estimator MDPDE are controlled by a tuning parameter that can be estimated adaptively base on a given data set. The consistency and the asymptotic distribution of the  MDPDE are theoretically derived. The influence function of the estimator indicates a redescending effect for vertical outliers. Our simulation studies show that the MDPDE with the optimum tuning parameter produces results almost as efficient as the OLS and GLS methods in pure data, while at the same time, it outperforms even the existing robust techniques in the presence of outliers. The real data example also confirms that the MDPDE gives an excellent fit to the major part of the data, whereas the OLS and GLS fits move towards the outliers.

	\bibliography{DPD_Reference}

\begin{thebibliography}{}

\bibitem[Aquaro and Cizek, 2013]{Aquaro2013}
Aquaro, M. and Cizek, P. (2013).
\newblock One-step robust estimation of fixed-effects panel data models.
\newblock {\em Computational Statistics and Data Analysis}, 57(1):536--548.

\bibitem[Athey et~al., 2021]{athey2021matrix}
Athey, S., Bayati, M., Doudchenko, N., Imbens, G., and Khosravi, K. (2021).
\newblock Matrix completion methods for causal panel data models.
\newblock {\em Journal of the American Statistical Association}, pages 1--15.

\bibitem[Bakar and Midi, 2015]{Bakar2015}
Bakar, N. M.~A. and Midi, H. (2015).
\newblock Robust centering in the fixed effect panel data model.
\newblock {\em Pakistan Journal of Statistics}, 31(1):33--48.

\bibitem[Balestra and Nerlove, 1966]{Balestra1966}
Balestra, P. and Nerlove, M. (1966).
\newblock Pooling cross-section and time series data in the estimation of a
  dynamic model: the demand for natural gas.
\newblock {\em Econometrica}, 34(3):585--612.

\bibitem[Baltagi, 2005]{Baltagi2005}
Baltagi, B.~H. (2005).
\newblock {\em Econometric Analysis of Panel Data}.
\newblock John Wiley and Sons, Chichester.

\bibitem[Basu et~al., 1998]{MR1665873}
Basu, A., Harris, I.~R., Hjort, N.~L., and Jones, M.~C. (1998).
\newblock Robust and efficient estimation by minimising a density power
  divergence.
\newblock {\em Biometrika}, 85(3):549--559.

\bibitem[Beyaztas and Bandyopadhyay, 2020]{Beyaztas2020}
Beyaztas, B.~H. and Bandyopadhyay, S. (2020).
\newblock Robust estimation for linear panel data models.
\newblock {\em Statistics in Medicine}, 39(29):4421--4438.

\bibitem[Bramati and Croux, 2007]{Bramati2007}
Bramati, M.~C. and Croux, C. (2007).
\newblock Robust estimators for the fixed effects panel data model.
\newblock {\em Econometrics Journal}, 10(3):521--540.

\bibitem[Cizek, 2010]{Cizek2010}
Cizek, P. (2010).
\newblock Reweighted least trimmed squares: an alternative to one-step
  estimators.
\newblock {\em CentER Discussion Paper Series 91/2010}.

\bibitem[Cox and Hall, 2002]{Cox2002}
Cox, D.~R. and Hall, P. (2002).
\newblock Estimation in a simple random effects model with nonnormal
  distributions.
\newblock {\em Biometrika}, 89(4):831--840.

\bibitem[Diggle et~al., 2002]{Diggle2002}
Diggle, P.~J., Heagerty, P., Liang, K.-Y., and Zeger, S.~L. (2002).
\newblock {\em Analysis of Longitudinal Data}.
\newblock Oxford University Press, United Kingdom.

\bibitem[Ferguson, 1996]{MR1699953}
Ferguson, T.~S. (1996).
\newblock {\em A course in large sample theory}.
\newblock Texts in Statistical Science Series. Chapman \& Hall, London.

\bibitem[Fitzmaurice et~al., 2004]{Fitzmaurice2004}
Fitzmaurice, G.~M., Laird, N.~M., and Ware, J.~H. (2004).
\newblock {\em Applied Longitudinal Analysis}.
\newblock John Wiley and Sons, New York.

\bibitem[Gardiner et~al., 2009]{Gardiner2009}
Gardiner, J.~C., Luo, Z., and Roman, L.~A. (2009).
\newblock Fixed effects, random effects and gee: What are the differences?
\newblock {\em Stat Med}, 28(2):221--239.

\bibitem[Gervini and Yohai, 2002]{Gervini2002}
Gervini, D. and Yohai, V.~J. (2002).
\newblock A class of robust and fully efficient regression estimators.
\newblock {\em The Annals of Statistics}, 30(2):583--616.

\bibitem[Ghosh and Basu, 2013]{MR3117102}
Ghosh, A. and Basu, A. (2013).
\newblock Robust estimation for independent non-homogeneous observations using
  density power divergence with applications to linear regression.
\newblock {\em Electron. J. Stat.}, 7:2420--2456.

\bibitem[Greene, 2017]{Greene2017}
Greene, W.~H. (2017).
\newblock {\em Econometric Analysis}.
\newblock Prentice Hall, New York.

\bibitem[Hampel et~al., 1986]{MR829458}
Hampel, F.~R., Ronchetti, E.~M., Rousseeuw, P.~J., and Stahel, W.~A. (1986).
\newblock {\em Robust statistics: The approach based on influence functions}.
\newblock John Wiley \& Sons, Inc., New York.

\bibitem[Hsiao, 1985]{Hsiao1985}
Hsiao, C. (1985).
\newblock Benefits and limitations of panel data.
\newblock {\em Econometric Reviews}, 4(1):121--174.

\bibitem[Hsiao, 2007]{Hsiao2007}
Hsiao, C. (2007).
\newblock Panel data analysis -- advantages and challenges.
\newblock {\em Test}, 16(1):1--22.

\bibitem[Huber, 1981]{MR606374}
Huber, P.~J. (1981).
\newblock {\em Robust statistics}.
\newblock John Wiley \& Sons, Inc., New York.
\newblock Wiley Series in Probability and Mathematical Statistics.

\bibitem[Jirata et~al., 2014]{Jirata2014}
Jirata, M.~T., Chelule, J.~C., and Odhiambo, R.~O. (2014).
\newblock Deriving some estimators of panel data regression models with
  individual effects.
\newblock {\em Int J Sci Res}, 3(5):53--59.

\bibitem[Kennedy, 2003]{Kennedy2003}
Kennedy, P. (2003).
\newblock {\em A Guide to Econometrics}.
\newblock The MIT Press, Cambridge.

\bibitem[Kutner et~al., 2004]{Kutner2004}
Kutner, M.~H., Nachtsheim, C.~J., and Neter, J. (2004).
\newblock {\em Applied Linear Regression Models}.
\newblock McGraw-Hill Education, New York.

\bibitem[Laird and Ware, 1982]{Laird1982}
Laird, N.~M. and Ware, J.~H. (1982).
\newblock Random-effects models for longitudinal data.
\newblock {\em Biometrics}, 38(4):963--974.

\bibitem[Lehmann, 1999]{MR1663158}
Lehmann, E.~L. (1999).
\newblock {\em Elements of large-sample theory}.
\newblock Springer Texts in Statistics. Springer-Verlag, New York.

\bibitem[Maddala and Mount, 1973]{Maddala1973}
Maddala, G.~S. and Mount, T.~D. (1973).
\newblock A comparative study of alternative estimators for variance components
  models used in econometric applications.
\newblock {\em Journal of the American Statistical Association},
  68(342):324--328.

\bibitem[Maronna et~al., 2006]{Maronna2006}
Maronna, R.~A., Martin, R.~D., and Yohai, V.~J. (2006).
\newblock {\em Robust Statistics. Theory and Methods}.
\newblock John Wiley and Sons, New York.

\bibitem[Maronna and Yohai, 2000]{Maronna2000}
Maronna, R.~A. and Yohai, V.~J. (2000).
\newblock Robust regression with both continuous and categorical predictors.
\newblock {\em Journal of Statistical Planning and Inference},
  89(1-2):197--214.

\bibitem[Midi and Muhammad, 2018]{Midi2018}
Midi, H. and Muhammad, S. (2018).
\newblock Robust estimation for fixed and random effects panel data models with
  different centering methods.
\newblock {\em Journal of Engineering and Applied Sciences}, 13(17):7156--7161.

\bibitem[Mundlak, 1978]{Mundlak1978}
Mundlak, Y. (1978).
\newblock On the pooling of time series and cross section data.
\newblock {\em Econometrica}, 46(1):69--85.

\bibitem[Rousseeuw, 1984]{Rousseeuw1984}
Rousseeuw, P.~J. (1984).
\newblock Least median of squares regression.
\newblock {\em J Am Stat Assoc}, 79(388):871--880.

\bibitem[Rousseeuw and Leroy, 2003]{Rousseeuw2003}
Rousseeuw, P.~J. and Leroy, A.~M. (2003).
\newblock {\em Robust Regression and Outlier Detection}.
\newblock John Wiley and Sons, New York.

\bibitem[Rousseeuw and van Zomeren, 1990]{Rousseeuw1990}
Rousseeuw, P.~J. and van Zomeren, B.~C. (1990).
\newblock Unmasking multivariate outliers and leverage points.
\newblock {\em Journal of the American Statistical Association},
  85(41):633--639.

\bibitem[Visek, 2015]{Visek2015}
Visek, J.~A. (2015).
\newblock Estimating the model with fixed and random effects by a robust
  method.
\newblock {\em Methodology and Computing in Applied Probability},
  17(4):999--1014.

\bibitem[Wallace and Hussain, 1969]{Wallace1969}
Wallace, T.~D. and Hussain, A. (1969).
\newblock The use of error components models in combining cross section and
  time-series data.
\newblock {\em Econometrica}, 37(1):55--72.

\bibitem[Warwick and Jones, 2005]{Warjones}
Warwick, J. and Jones, M. (2005).
\newblock Choosing a robustness tuning parameter.
\newblock {\em J. Statist. Comput. Simulation}, 75(7):581--588.

\end{thebibliography}
	
	\section*{Appendix}
	\addcontentsline{toc}{section}{Appendices}
	\renewcommand{\thesubsection}{\Alph{subsection}}
	\numberwithin{equation}{subsection}
	
	\subsection{Proof of Theorem \ref{theorem:asymp}}
	\begin{proof}
		The proof of the first part closely follows the consistency of the maximum likelihood estimator with the line of modifications as given in Theorem 3.1 of \cite{MR3117102}. For brevity, we only present the detailed proof of the second part. 
		
		Let $\widehat{\theta}$ be the MDPDE of $\theta$. Then
		\begin{equation}
			\frac{\partial}{\partial \theta} \widehat{d}_\gamma(f_\theta, g) =  \frac{\partial}{\partial \theta} \left[\frac{1}{N} \sum_{i=1}^{N}	\int_y  f^{1+\gamma}_\theta(y|x_i) dy - \frac{1+\gamma}{N \gamma} \sum_{i=1}^{N} f_\theta^{\gamma}(y_i | x_i)\right] =0.
		\end{equation}
		Thus, it can be written as the estimating equation of an M-estimator as follows
		\begin{equation}
			\sum_{i=1}^N  \Psi_{\widehat{\theta}}(y_i|x_i) = 0, 
			\label{psi}
		\end{equation}
		where
		\begin{equation}
			\Psi_{\theta}(y_i|x_i) = u_\theta(y_i|x_i) f_\theta^\gamma(y_i|x_i) - \int_y u_\theta(y|x_i) f_\theta^{1+\gamma}(y_i|x_i) dy.
		\end{equation}
		Let $\theta_g$ be the true value of $\theta$, then $ E\left(\sum_{i=1}^N  \Psi_{\widehat{\theta}}(y_i|x_i)\right) = 0$ gives
		\begin{equation}
			\sum_{i=1}^N \Bigg[\int_y u_{\theta_g}(y|x_i) f_{\theta_g}^\gamma(y|x_i) g(y|x_i) dy  - \int_y u_{\theta_g}(y|x_i) f_{\theta_g}^{1+\gamma}(y_i|x_i) dy \Bigg] = 0 .    \label{theta_g}
		\end{equation}
		Taking a Taylor series expansion of Equation \eqref{psi}, we get
		\begin{equation}
			\begin{split}
				\frac{1}{N}\sum_{i=1}^N  \Psi_{\theta_g}(y_i|x_i)  &+ \frac{1}{N}\sum_{i=1}^N  \frac{\partial }{\partial \theta}\Psi_\theta(y_i|x_i)\Big|_{\theta = \theta_g} (\widehat{\theta} -\theta_g) + R_N = 0,\\
				\mbox{or } \sqrt{N}(\widehat{\theta} -\theta_g) & = - \left[ \frac{1}{N}\sum_{i=1}^N  \frac{\partial }{\partial \theta}\Psi_\theta(y_i|x_i)\Big|_{\theta = \theta_g} \right]^{-1} \left[\frac{1}{\sqrt{N}}\sum_{i=1}^N  \Psi_{\theta_g}(y_i|x_i) + \sqrt{N}R_N\right],
				\label{taylor}
			\end{split}
		\end{equation}
		where $R_N$ is the remainder  term. 
		Using the weak law of large numbers (WLLN), we have
		\begin{equation}
			\begin{split}
				& \frac{1}{N}\sum_{i=1}^N  \frac{\partial }{\partial \theta}\Psi_\theta(y_i|x_i) \\
				& \overset{p}{\to} \lim_{N \rightarrow \infty} E\left[ \frac{1}{N}\sum_{i=1}^N  \frac{\partial }{\partial \theta}\Psi_\theta(y_i|x_i) \right]\\
				& \overset{p}{\to} \lim_{N \rightarrow \infty} \frac{1}{N}\sum_{i=1}^N  E\left[ \frac{\partial }{\partial \theta} \left( u_\theta f_\theta^\gamma -  \int u_\theta f_\theta^{1+\gamma}  \right) \right]\\
				& \overset{p}{\to} \lim_{N \rightarrow \infty} \frac{1}{N}\sum_{i=1}^N  E\left[  -I_\theta f_\theta^\gamma + \gamma u_\theta u_\theta^T f_\theta^\gamma -  \int \left\{ -  I_\theta f_\theta^{1+\gamma} + (1+\gamma) u_\theta u_\theta^T f_\theta^{1+\gamma}  \right\} \right]\\
				& \overset{p}{\to}  \lim_{N \rightarrow \infty} \frac{1}{N}\sum_{i=1}^N  \left[ -  \int I_\theta f_\theta^\gamma g + \gamma \int u_\theta u_\theta^T f_\theta^\gamma g +  \int    I_\theta f_\theta^{1+\gamma} - (1+\gamma) \int u_\theta u_\theta^T f_\theta^{1+\gamma}  \right]\\
				& \overset{p}{\to} - \lim_{N \rightarrow \infty} \frac{1}{N}\sum_{i=1}^N  \left[  \int u_\theta u_\theta^T f_\theta^{1+\gamma} +  \int \left(I_\theta -\gamma  u_\theta u_\theta^T \right) (g - f_\theta) f_\theta^\gamma   \right].
			\end{split}
		\end{equation}
		So
		\begin{equation}
			\frac{1}{N}\sum_{i=1}^N  \frac{\partial }{\partial \theta}\Psi_\theta(y_i|x_i)\Big|_{\theta = \theta_g}  \overset{p}{\to} - \lim_{N \rightarrow \infty} \frac{1}{N}\sum_{i=1}^N \bJ^{(i)} = -  \bJ .
			\label{bj}
		\end{equation}
		From Equation \eqref{theta_g}, we get
		\begin{equation}
			\begin{split}
				& E\left[\frac{1}{\sqrt{N}}\sum_{i=1}^N  \Psi_{\theta_g}(y_i|x_i) \right] \\
				& = \frac{1}{\sqrt{N}} \sum_{i=1}^N \Bigg[\int_y u_{\theta_g}(y|x_i) f_{\theta_g}^\gamma(y|x_i) g(y|x_i) dy  - \int_y u_{\theta_g}(y|x_i) f_{\theta_g}^{1+\gamma}(y_i|x_i) dy \Bigg]\\
				&= 0.     \label{E_phi}
			\end{split}
		\end{equation}
		Now
		\begin{equation}
			\begin{split}
				& V\left[\frac{1}{\sqrt{N}}\sum_{i=1}^N  \Psi_{\theta_g}(y_i|x_i) \right] \\
				& = \frac{1}{N} \sum_{i=1}^N V\left[ \Psi_{\theta_g}(y_i|x_i) \right] \\
				& = \frac{1}{N} \sum_{i=1}^N \Bigg[\int_y u_{\theta_g}(y|x_i) u_{\theta_g}^T(y|x_i) f_{\theta_g}^{2\gamma}(y|x_i) g(y|x_i) dy  - \xi^{(i)} \xi^{(i)T}\Bigg]\\
				& = \frac{1}{N} \sum_{i=1}^N \bK^{(i)}.   
				\label{var_phi}
			\end{split}
		\end{equation}
		Following Section 5 of \cite{MR1699953} or Section 2.7 of \cite{MR1663158} and using Equations \eqref{E_phi} and \eqref{var_phi}, the central limit theorem (CLT) for the independent but not identical random variables gives
		\begin{equation}
			\frac{1}{\sqrt{N}}\sum_{i=1}^N  \Psi_{\theta_g}(y_i|x_i) \overset{a}{\sim} N\left(0, \bK\right).
			\label{clt}
		\end{equation}
		Under regularity condition (A3), it can be easily shown that the reminder term $\sqrt{N}R_N = o_p(1).$ 
		Therefore, combining Equations \eqref{bj} and \eqref{clt}, we get from Equation \eqref{taylor}
		\begin{equation}
			\sqrt{N}(\widehat{\theta} -\theta_g) \overset{a}{\sim} N\left(0, \bJ^{-1}\bK \bJ^{-1}\right).
		\end{equation}
		This completes the proof.
	\end{proof}

	\subsection{Vector $\xi^{(i)}$ at Model} \label{appendix:xi}
	
	From Equations \eqref{score_functions} and \eqref{J_i}, we get
	\begin{equation}
		\begin{split}
			\xi_\beta^{(i)} &= \int_{y_i} 	u_\beta(y_i | x_i)  f_\theta^{1+\gamma}(y_i | x_i) dy_i \\
			& = \int_{y_i}	\left[ \frac{1}{\sigma_\epsilon^2}\sum_{t=1}^{T} x_{it} (y_{it} - x_{it} \beta) - \frac{ T \bar{x}_i \sigma_\alpha^2 } {\sigma_\epsilon^2 ( \sigma_\epsilon^2 + T \sigma_\alpha^2 )}  \sum_{t=1}^{T} (y_{it} - x_{it} \beta)\right] f_\theta^{1+\gamma}(y_i | x_i) dy_i \\
			& = 0, \ \ \mbox{from \eqref{int_sr_ss}}.
			\label{xi_beta}
		\end{split}
	\end{equation}	
	
	\noindent
	From Equations \eqref{score_functions} and \eqref{J_i}, we get
	\begin{equation}
		\begin{split}
			\xi_{\sigma_\alpha^2}^{(i)} &= \int_{y_i} 	u_{\sigma_\alpha^2}(y | x_i)  f_\theta^{1+\gamma}(y_i | x_i) dy_i \\
			& = \int_{y_i}	\left[ -\frac{T}{2  ( \sigma_\epsilon^2 + T \sigma_\alpha^2 )} 
			+ \frac{1} {2( \sigma_\epsilon^2 + T \sigma_\alpha^2 )^2} \left\{ \sum_{t=1}^{T} (y_{it} - x_{it} \beta) \right\}^2\right] f_\theta^{1+\gamma}(y_i | x_i) dy_i\\
			& =  -\frac{T}{2  ( \sigma_\epsilon^2 + T \sigma_\alpha^2 )} \times M(1+\gamma) \mbox{, using }\eqref{int_f}\\
			& \hspace{1in} + \frac{1} {2( \sigma_\epsilon^2 + T \sigma_\alpha^2 )^2} \times MT ( \sigma_\epsilon^2  + T \sigma_\alpha^2) \mbox{, using \eqref{s2}}\\
			& =  -\frac{MT \gamma}{2  ( \sigma_\epsilon^2 + T \sigma_\alpha^2 )} .
			\label{xi_alpha}
		\end{split}
	\end{equation}	
	
	From Equations \eqref{score_functions} and \eqref{J_i}, we get
	\begin{equation}
		\begin{split}
			\xi_{\sigma_\epsilon^2}^{(i)} &= \int_{y_i} 	u_{\sigma_\epsilon^2}(y | x_i)  f_\theta^{1+\gamma}(y_i | x_i) dy_i \\
			& = \int_{y_i}	\Bigg[ -\frac{T \left[  \sigma_\epsilon^2 + (T-1) \sigma_\alpha^2 \right]}{2 \sigma_\epsilon^{2} ( \sigma_\epsilon^2 + T \sigma_\alpha^2 )} 	 + \frac{1}{2\sigma_\epsilon^4}\sum_{t=1}^{T} (y_{it} - x_{it} \beta)^2 \\
			& \hspace{2cm} - \frac{\sigma_\alpha^2 (2\sigma_\epsilon^2 + T \sigma_\alpha^2 )} {2\sigma_\epsilon^4 ( \sigma_\epsilon^2 + T \sigma_\alpha^2 )^2} \left\{ \sum_{t=1}^{T} (y_{it} - x_{it} \beta) \right\}^2 \Bigg] f_\theta^{1+\gamma}(y_i | x_i) dy_i\\
			& = -\frac{T \left[  \sigma_\epsilon^2 + (T-1) \sigma_\alpha^2 \right]}{2 \sigma_\epsilon^{2} ( \sigma_\epsilon^2 + T \sigma_\alpha^2 )}  \times M(1+\gamma) \mbox{, using }\eqref{int_f}\\
			& \hspace{1in} + \frac{1}{2\sigma_\epsilon^4} \times T M ( \sigma_\epsilon^2  + \sigma_\alpha^2 ) \mbox{, using } \eqref{int_y2}\\
			& \hspace{1in} - \frac{\sigma_\alpha^2 (2\sigma_\epsilon^2 + T \sigma_\alpha^2 )} {2\sigma_\epsilon^4 ( \sigma_\epsilon^2 + T \sigma_\alpha^2 )^2} \times MT ( \sigma_\epsilon^2  + T \sigma_\alpha^2) \mbox{, using \eqref{s2}}\\
			& = -\frac{M T \gamma \left[  \sigma_\epsilon^2 + (T-1) \sigma_\alpha^2 \right]}{2 \sigma_\epsilon^2 ( \sigma_\epsilon^2 + T \sigma_\alpha^2 ) }.
			\label{xi_epsilon}
		\end{split}
	\end{equation}

	\subsection{Matrix $\bJ^{(i)}$ at Model} \label{appendix:J}

	From Equations \eqref{score_functions} and \eqref{J_i}, we get
	\begin{equation}
		\begin{split}
			\bJ_{\sigma_\alpha^2}^{(i)} &= \int_{y_i} 	u^2_{\sigma_\alpha^2}(y | x_i)  f_\theta^{1+\gamma}(y_i | x_i) dy_i \\
			& = \int_{y_i}	\left[ -\frac{T}{2  ( \sigma_\epsilon^2 + T \sigma_\alpha^2 )} 
			+ \frac{1} {2( \sigma_\epsilon^2 + T \sigma_\alpha^2 )^2} \left\{ \sum_{t=1}^{T} (y_{it} - x_{it} \beta) \right\}^2\right]^2 f_\theta^{1+\gamma}(y_i | x_i) dy_i\\
			&= 	\frac{T^2}{4  ( \sigma_\epsilon^2 + T \sigma_\alpha^2 )^2} \int_{y_i} f_\theta^{1+\gamma}(y_i | x_i) dy_i	
			- \frac{T}{2  ( \sigma_\epsilon^2 + T \sigma_\alpha^2 )^3} \int_{y_i} \left\{ \sum_{t=1}^{T} (y_{it} - x_{it} \beta) \right\}^2 f_\theta^{1+\gamma}(y_i | x_i) dy_i\\
			& \hspace{1cm} +
			\frac{1}{4 ( \sigma_\epsilon^2 + T \sigma_\alpha^2 )^4} \int_{y_i} \left\{ \sum_{t=1}^{T} (y_{it} - x_{it} \beta) \right\}^4 f_\theta^{1+\gamma}(y_i | x_i) dy_i\\
			& =  M (1+\gamma) \frac{T^2}{4  ( \sigma_\epsilon^2 + T \sigma_\alpha^2 )^2}   \mbox{, using \eqref{int_f})}\\
			& \ \ \ \ - \frac{T}{2  ( \sigma_\epsilon^2 + T \sigma_\alpha^2 )^3} \times MT( \sigma_\epsilon^2  + T \sigma_\alpha^2)    \mbox{, using \eqref{s2}) }\\
			& \ \ \ \ + \frac{1}{4 ( \sigma_\epsilon^2 + T \sigma_\alpha^2 )^4} \times \frac{3MT^2}{(1+\gamma)}( \sigma_\epsilon^2  + T \sigma_\alpha^2)^2    \mbox{, using \eqref{s4})}\\
			& = \frac{M T^2 ( \gamma^2 + 2 )} {4 (1+\gamma) ( \sigma_\epsilon^2 + T \sigma_\alpha^2 )^2} .
		\end{split}
	\end{equation}	
	
	\noindent
	From Equations \eqref{score_functions} and \eqref{J_i}, we get
	\begin{equation}
		\begin{split}
			\bJ_\beta^{(i)} &= \int_{y_i} 	u_\beta(y_i | x_i) u^T_\beta(y_i | x_i) f_\theta^{1+\gamma}(y_i | x_i) dy_i \\
			& = \int_{y_i}	\left[ \frac{1}{\sigma_\epsilon^2}\sum_{t=1}^{T} x_{it} (y_{it} - x_{it} \beta) - \frac{ T \bar{x}_i \sigma_\alpha^2 } {\sigma_\epsilon^2 ( \sigma_\epsilon^2 + T \sigma_\alpha^2 )}  \sum_{t=1}^{T} (y_{it} - x_{it} \beta)\right]\\
			& \hspace{1cm}	\left[ \frac{1}{\sigma_\epsilon^2}\sum_{t=1}^{T} x_{it} (y_{it} - x_{it} \beta) - \frac{ T \bar{x}_i \sigma_\alpha^2 } {\sigma_\epsilon^2 ( \sigma_\epsilon^2 + T \sigma_\alpha^2 )}  \sum_{t=1}^{T} (y_{it} - x_{it} \beta)\right]^T f_\theta^{1+\gamma}(y_i | x_i) dy_i\\
			& = \sum_{t=1}^{T} \Bigg[\frac{1}{\sigma_\epsilon^4} x_{it} x_{it}^T  -  \frac{ 2  T \sigma_\alpha^2 x_{it}  \bar{x}_i^T } {\sigma_\epsilon^4 ( \sigma_\epsilon^2 + T \sigma_\alpha^2 )} + \frac{ T^2  \sigma_\alpha^4 \bar{x}_i \bar{x}_i^T} {\sigma_\epsilon^4 ( \sigma_\epsilon^2 + T \sigma_\alpha^2 )^2} \Bigg]\int_{y_i} (y_{it} - x_{it} \beta)^2 f_\theta^{1+\gamma}(y_i | x_i) dy_i\\
			& \hspace{1cm} + \sum_{t\neq t'} \Bigg[\frac{1}{\sigma_\epsilon^4} x_{it} x_{it'}^T  -  \frac{ 2  T \sigma_\alpha^2 x_{it}  \bar{x}_i^T } {\sigma_\epsilon^4 ( \sigma_\epsilon^2 + T \sigma_\alpha^2 )} + \frac{ T^2  \sigma_\alpha^4 \bar{x}_i \bar{x}_i^T} {\sigma_\epsilon^4 ( \sigma_\epsilon^2 + T \sigma_\alpha^2 )^2} \Bigg]\int_{y_i} (y_{it} - x_{it} \beta)(y_{it'} - x_{it'} \beta) f_\theta^{1+\gamma}(y_i | x_i) dy_i\\
			&  = M ( \sigma_\epsilon^2  + \sigma_\alpha^2 )\sum_{t=1}^{T} \Bigg[\frac{1}{\sigma_\epsilon^4} x_{it} x_{it}^T  -  \frac{ 2  T \sigma_\alpha^2 x_{it}  \bar{x}_i^T } {\sigma_\epsilon^4 ( \sigma_\epsilon^2 + T \sigma_\alpha^2 )} + \frac{ T^2  \sigma_\alpha^4 \bar{x}_i \bar{x}_i^T} {\sigma_\epsilon^4 ( \sigma_\epsilon^2 + T \sigma_\alpha^2 )^2} \Bigg] \mbox{, using }\eqref{int_y2}\\
			& \hspace{1cm} +  M  \sigma_\alpha^2 \sum_{t\neq t'} \Bigg[\frac{1}{\sigma_\epsilon^4} x_{it} x_{it'}^T  -  \frac{ 2  T \sigma_\alpha^2 x_{it}  \bar{x}_i^T } {\sigma_\epsilon^4 ( \sigma_\epsilon^2 + T \sigma_\alpha^2 )} + \frac{ T^2  \sigma_\alpha^4 \bar{x}_i \bar{x}_i^T} {\sigma_\epsilon^4 ( \sigma_\epsilon^2 + T \sigma_\alpha^2 )^2} \Bigg]  \mbox{, using }\eqref{int_yy}\\
			& = M  \sigma_\alpha^2 \sum_{t, t' = 1}^T \Bigg[\frac{1}{\sigma_\epsilon^4} x_{it} x_{it'}^T  -  \frac{ 2  T \sigma_\alpha^2 x_{it}  \bar{x}_i^T } {\sigma_\epsilon^4 ( \sigma_\epsilon^2 + T \sigma_\alpha^2 )} + \frac{ T^2  \sigma_\alpha^4 \bar{x}_i \bar{x}_i^T} {\sigma_\epsilon^4 ( \sigma_\epsilon^2 + T \sigma_\alpha^2 )^2} \Bigg]\\
			&  \hspace{1cm} + M  \sigma_\epsilon^2 \sum_{t=1}^{T} \Bigg[\frac{1}{\sigma_\epsilon^4} x_{it} x_{it}^T  -  \frac{ 2  T \sigma_\alpha^2 x_{it}  \bar{x}_i^T } {\sigma_\epsilon^4 ( \sigma_\epsilon^2 + T \sigma_\alpha^2 )} + \frac{ T^2  \sigma_\alpha^4 \bar{x}_i \bar{x}_i^T} {\sigma_\epsilon^4 ( \sigma_\epsilon^2 + T \sigma_\alpha^2 )^2} \Bigg]\\
			&  = M  \sigma_\alpha^2  \Bigg[\frac{T^2}{\sigma_\epsilon^4} \bar{x}_i \bar{x}_i^T  -  \frac{ 2  T^3 \sigma_\alpha^2 \bar{x}_i \bar{x}_i^T } {\sigma_\epsilon^4 ( \sigma_\epsilon^2 + T \sigma_\alpha^2 )} + \frac{ T^4  \sigma_\alpha^4 \bar{x}_i \bar{x}_i^T} {\sigma_\epsilon^4 ( \sigma_\epsilon^2 + T \sigma_\alpha^2 )^2} \Bigg],  \\
			&  \hspace{1cm} + M  \sigma_\epsilon^2 \sum_{t=1}^{T} \frac{1}{\sigma_\epsilon^4} x_{it} x_{it}^T   + M  \sigma_\epsilon^2  \Bigg[  -  \frac{ 2  T^2 \sigma_\alpha^2 \bar{x}_i \bar{x}_i^T} {\sigma_\epsilon^4 ( \sigma_\epsilon^2 + T \sigma_\alpha^2 )} + \frac{ T^3  \sigma_\alpha^4 \bar{x}_i \bar{x}_i^T} {\sigma_\epsilon^4 ( \sigma_\epsilon^2 + T \sigma_\alpha^2 )^2} \Bigg]\\
			& = M  \Bigg[  \frac{\sigma_\epsilon^2}{\sigma_\epsilon^4}  \sum_{t=1}^{T} x_{it} x_{it}^T +  \frac{T^2 \sigma_\alpha^2 }{\sigma_\epsilon^4} \bar{x}_i \bar{x}_i^T + (\sigma_\epsilon^2 + T \sigma_\alpha^2) \left\{  -  \frac{ 2 T^2 \sigma_\alpha^2 \bar{x}_i \bar{x}_i^T} {\sigma_\epsilon^4 ( \sigma_\epsilon^2 + T \sigma_\alpha^2 )} + \frac{ T^3  \sigma_\alpha^4 \bar{x}_i \bar{x}_i^T} {\sigma_\epsilon^4 ( \sigma_\epsilon^2 + T \sigma_\alpha^2 )^2} \right\} \Bigg]\\
			&  = M \sigma_\epsilon^{-4}   \Bigg[  \sigma_\epsilon^2 \sum_{t=1}^{T} x_{it} x_{it}^T +  T^2 \sigma_\alpha^2  \bar{x}_i \bar{x}_i^T    - 2 T^2   \sigma_\alpha^2 \bar{x}_i \bar{x}_i^T +  \frac{T^3  \sigma_\alpha^4 \bar{x}_i \bar{x}_i^T}{( \sigma_\epsilon^2 + T \sigma_\alpha^2 )}  \Bigg]\\
			& = M \sigma_\epsilon^{-4} \Bigg[ \sigma_\epsilon^2 \sum_{t=1}^{T} x_{it} x_{it}^T   +  T^2 \sigma_\alpha^2 \left( \frac{ T \sigma_\alpha^2 }{( \sigma_\epsilon^2 + T \sigma_\alpha^2 )} -1 \right) \bar{x}_i \bar{x}_i^T  \Bigg] .
		\end{split}
	\end{equation}

	From Equations \eqref{score_functions} and \eqref{J_i}, we get
	\begin{equation}
		\begin{split}
			\bJ_{\sigma_\epsilon^2}^{(i)} &= \int_{y_i} 	u_{\sigma_\epsilon^2}^2(y | x_i)  f_\theta^{1+\gamma}(y_i | x_i) dy_i \\
			& = \int_{y_i}	\Bigg[ -\frac{T \left[  \sigma_\epsilon^2 + (T-1) \sigma_\alpha^2 \right]}{2 \sigma_\epsilon^{2} ( \sigma_\epsilon^2 + T \sigma_\alpha^2 )} 	 + \frac{1}{2\sigma_\epsilon^4}\sum_{t=1}^{T} (y_{it} - x_{it} \beta)^2 \\
			& \hspace{2cm} - \frac{\sigma_\alpha^2 (2\sigma_\epsilon^2 + T \sigma_\alpha^2 )} {2\sigma_\epsilon^4 ( \sigma_\epsilon^2 + T \sigma_\alpha^2 )^2} \left\{ \sum_{t=1}^{T} (y_{it} - x_{it} \beta) \right\}^2 \Bigg]^2 f_\theta^{1+\gamma}(y_i | x_i) dy_i\\
			&= 	\frac{T^2 \left[  \sigma_\epsilon^2 + (T-1) \sigma_\alpha^2 \right]^2}{4 \sigma_\epsilon^{4} ( \sigma_\epsilon^2 + T \sigma_\alpha^2 )^2} \int_{y_i} f_\theta^{1+\gamma}(y_i | x_i) dy_i	\\
			& \hspace{1cm} + \frac{1}{4\sigma_\epsilon^8} \int_{y_i} \left\{ \sum_{t=1}^{T} (y_{it} - x_{it} \beta)^2 \right\}^2 f_\theta^{1+\gamma}(y_i | x_i) dy_i\\
			& \hspace{1cm} + \frac{\sigma_\alpha^4 (2\sigma_\epsilon^2 + T \sigma_\alpha^2 )^2} {4\sigma_\epsilon^8 ( \sigma_\epsilon^2 + T \sigma_\alpha^2 )^4} \int_{y_i} \left\{ \sum_{t=1}^{T} (y_{it} - x_{it} \beta) \right\}^4 f_\theta^{1+\gamma}(y_i | x_i) dy_i\\
			& \hspace{1cm} -
			\frac{T \left[  \sigma_\epsilon^2 + (T-1) \sigma_\alpha^2 \right]}{2\sigma_\epsilon^6 ( \sigma_\epsilon^2 + T \sigma_\alpha^2 )}  \sum_{t=1}^{T}  \int_{y_i}(y_{it} - x_{it} \beta)^2  f_\theta^{1+\gamma}(y_i | x_i) dy_i\\
			& \hspace{1cm} + \frac{\sigma_\alpha^2 (2\sigma_\epsilon^2 + T \sigma_\alpha^2 ) T \left[  \sigma_\epsilon^2 + (T-1) \sigma_\alpha^2 \right]	} {2\sigma_\epsilon^6 ( \sigma_\epsilon^2 + T \sigma_\alpha^2 )^3} \int_{y_i} \left\{ \sum_{t=1}^{T} (y_{it} - x_{it} \beta) \right\}^2 f_\theta^{1+\gamma}(y_i | x_i) dy_i\\	
			& \hspace{1cm}  - \frac{\sigma_\alpha^2 (2\sigma_\epsilon^2 + T \sigma_\alpha^2 )} {2\sigma_\epsilon^8 ( \sigma_\epsilon^2 + T \sigma_\alpha^2 )^2} \sum_{t=1}^{T} \int_{y_i} (y_{it} - x_{it} \beta)^2 \left\{ \sum_{t'=1}^{T} (y_{it'} - x_{it'} \beta) \right\}^2 f_\theta^{1+\gamma}(y_i | x_i) dy_i \\
			&= \frac{T^2 \left[  \sigma_\epsilon^2 + (T-1) \sigma_\alpha^2 \right]^2}{4 \sigma_\epsilon^{4} ( \sigma_\epsilon^2 + T \sigma_\alpha^2 )^2} \times	M  (1+\gamma)   \mbox{, using \eqref{int_f}}\\
			& \hspace{1cm} + \frac{1}{4\sigma_\epsilon^8} \times M \left[T(T+2)  \sigma_\epsilon^4  +  2T(T+2)  \sigma_\epsilon^2\sigma_\alpha^2 + 3T^2 \sigma_\alpha^4 \right]     \mbox{, using \eqref{int_square}}\\
			& \hspace{1cm} + \frac{\sigma_\alpha^4 (2\sigma_\epsilon^2 + T \sigma_\alpha^2 )^2} {4\sigma_\epsilon^8 ( \sigma_\epsilon^2 + T \sigma_\alpha^2 )^4} \times  \frac{3MT^2}{(1+\gamma)} ( \sigma_\epsilon^2  + T \sigma_\alpha^2)^2    \mbox{, using \eqref{s4})}\\
			& \hspace{1cm} -
			\frac{T \left[  \sigma_\epsilon^2 + (T-1) \sigma_\alpha^2 \right]}{2\sigma_\epsilon^6 ( \sigma_\epsilon^2 + T \sigma_\alpha^2 )} \times T M ( \sigma_\epsilon^2  + \sigma_\alpha^2 )     \mbox{, using \eqref{int_y2}}\\
			& \hspace{1cm} + \frac{\sigma_\alpha^2 (2\sigma_\epsilon^2 + T \sigma_\alpha^2 ) T \left[  \sigma_\epsilon^2 + (T-1) \sigma_\alpha^2 \right]	} {2\sigma_\epsilon^6 ( \sigma_\epsilon^2 + T \sigma_\alpha^2 )^3} \times MT( \sigma_\epsilon^2  + T \sigma_\alpha^2)    \mbox{, using \eqref{s2}} \\	
			& \hspace{1cm}  - \frac{\sigma_\alpha^2 (2\sigma_\epsilon^2 + T \sigma_\alpha^2 )} {2\sigma_\epsilon^8 ( \sigma_\epsilon^2 + T \sigma_\alpha^2 )^2} \times T M (1+\gamma)\Big[ (T+2)  \sigma_\epsilon^4 + (T^2 + 2T +3) \sigma_\epsilon^2 \sigma_\alpha^2   \\
			& \hspace{2cm} + 3(T^2 -T + 1 ) \sigma_\alpha^4 \Big] \mbox{, using \eqref{E_s2_s2} }\\
			& = 	 \frac{  M T^2 (\gamma - 1) \left[  \sigma_\epsilon^2 + (T-1) \sigma_\alpha^2 \right]^2} {4\sigma_\epsilon^4  ( \sigma_\epsilon^2 + T \sigma_\alpha^2 )^2}   \\
			& \hspace{1cm} + \frac{MT}{4\sigma_\epsilon^8}   \left[(T+2)  \sigma_\epsilon^4  +  2(T+2)  \sigma_\epsilon^2\sigma_\alpha^2 + 3T \sigma_\alpha^4 \right]    + \frac{3MT^2\sigma_\alpha^4 (2\sigma_\epsilon^2 + T \sigma_\alpha^2 )^2} {4\sigma_\epsilon^8 (1+\gamma) ( \sigma_\epsilon^2 + T \sigma_\alpha^2 )^2}  \\
			& \hspace{1cm}  - \frac{ T M (1+\gamma)\sigma_\alpha^2 (2\sigma_\epsilon^2 + T \sigma_\alpha^2 )} {2\sigma_\epsilon^8 ( \sigma_\epsilon^2 + T \sigma_\alpha^2 )^2} \Big[ (T+2)  \sigma_\epsilon^4 + (T^2 + 2T +3) \sigma_\epsilon^2 \sigma_\alpha^2 + 3(T^2 -T + 1 ) \sigma_\alpha^4 \Big].
		\end{split}
	\end{equation}

	From Equations \eqref{score_functions} and \eqref{J_i}, we get
	\begin{equation}
		\begin{split}
			\bJ_{\beta, \ \sigma_\alpha^2}^{(i)} &= \int_{y_i} 	u_\beta(y_i | x_i) u_{ \sigma_\alpha^2}(y_i | x_i) f_\theta^{1+\gamma}(y_i | x_i) dy_i \\
			& = \int_{y_i}	\left[ \frac{1}{\sigma_\epsilon^2}\sum_{t=1}^{T} x_{it} (y_{it} - x_{it} \beta) - \frac{ T \bar{x}_i \sigma_\alpha^2 } {\sigma_\epsilon^2 ( \sigma_\epsilon^2 + T \sigma_\alpha^2 )}  \sum_{t=1}^{T} (y_{it} - x_{it} \beta)\right]\\
			& \hspace{1cm} \times \left[ -\frac{T}{2  ( \sigma_\epsilon^2 + T \sigma_\alpha^2 )} 
			+ \frac{1} {2( \sigma_\epsilon^2 + T \sigma_\alpha^2 )^2} \left\{ \sum_{t=1}^{T} (y_{it} - x_i \beta) \right\}^2\right] f_\theta^{1+\gamma}(y_i | x_i) dy_i\\
			& =  0 \mbox{ as all odd moments similar to \eqref{int_sr_ss}}	.	
		\end{split}
	\end{equation}	
	
	From Equations \eqref{score_functions} and \eqref{J_i}, we get
	\begin{equation}
		\begin{split}
			\bJ_{\beta, \ \sigma_\epsilon^2}^{(i)} &= \int_{y_i} 	u_\beta(y_i | x_i) u_{ \sigma_\epsilon^2}(y_i | x_i) f_\theta^{1+\gamma}(y_i | x_i) dy_i \\
			& = \int_{y_i}	\left[ \frac{1}{\sigma_\epsilon^2}\sum_{t=1}^{T} x_{it} (y_{it} - x_{it} \beta) - \frac{ T \bar{x}_i \sigma_\alpha^2 } {\sigma_\epsilon^2 ( \sigma_\epsilon^2 + T \sigma_\alpha^2 )}  \sum_{t=1}^{T} (y_{it} - x_{it} \beta)\right]\\
			& \hspace{1cm} \times 	\Bigg[ -\frac{T \left[  \sigma_\epsilon^2 + (T-1) \sigma_\alpha^2 \right]}{2 \sigma_\epsilon^{2} ( \sigma_\epsilon^2 + T \sigma_\alpha^2 )} 	 + \frac{1}{2\sigma_\epsilon^4}\sum_{t=1}^{T} (y_{it} - x_{it} \beta)^2 \\
			& \hspace{2cm} - \frac{\sigma_\alpha^2 (2\sigma_\epsilon^2 + T \sigma_\alpha^2 )} {2\sigma_\epsilon^4 ( \sigma_\epsilon^2 + T \sigma_\alpha^2 )^2} \left\{ \sum_{t=1}^{T} (y_{it} - x_{it} \beta) \right\}^2 \Bigg] f_\theta^{1+\gamma}(y_i | x_i) dy_i\\
			& =  0 \mbox{ as all odd moments similar to \eqref{int_sr_ss}}	.	
		\end{split}
	\end{equation}	
	
	From Equations \eqref{score_functions} and \eqref{J_i}, we get
	\begin{equation}
		\begin{split}
			\bJ_{\sigma_\alpha^2, \ \sigma_\epsilon^2}^{(i)} &= \int_{y_i} 	u_{\sigma_\alpha^2}(y_i | x_i) u_{ \sigma_\epsilon^2}(y_i | x_i) f_\theta^{1+\gamma}(y_i | x_i) dy_i \\
			& = \int_{y_i}\left[ -\frac{T}{2  ( \sigma_\epsilon^2 + T \sigma_\alpha^2 )} 
			+ \frac{1} {2( \sigma_\epsilon^2 + T \sigma_\alpha^2 )^2} \left\{ \sum_{t=1}^{T} (y_{it} - x_{it} \beta) \right\}^2\right]\\
			& \hspace{1cm} \times 	\Bigg[ -\frac{T \left[  \sigma_\epsilon^2 + (T-1) \sigma_\alpha^2 \right]}{2 \sigma_\epsilon^{2} ( \sigma_\epsilon^2 + T \sigma_\alpha^2 )} 	 + \frac{1}{2\sigma_\epsilon^4}\sum_{t=1}^{T} (y_{it} - x_{it} \beta)^2 \\
			& \hspace{2cm} - \frac{\sigma_\alpha^2 (2\sigma_\epsilon^2 + T \sigma_\alpha^2 )} {2\sigma_\epsilon^4 ( \sigma_\epsilon^2 + T \sigma_\alpha^2 )^2} \left\{ \sum_{t=1}^{T} (y_{it} - x_{it} \beta) \right\}^2 \Bigg] f_\theta^{1+\gamma}(y_i | x_i) dy_i\\
			& =  \frac{T^2 \left[  \sigma_\epsilon^2 + (T-1) \sigma_\alpha^2 \right]}{4 \sigma_\epsilon^{2} ( \sigma_\epsilon^2 + T \sigma_\alpha^2 )^2} \int_{y_i} f_\theta^{1+\gamma}(y_i | x_i) dy_i\\
			& \hspace{1cm}  + \frac{1} {4 \sigma_\epsilon^4 ( \sigma_\epsilon^2 + T \sigma_\alpha^2 )^2} \sum_{t=1}^{T} \int_{y_i} (y_{it} - x_{it} \beta)^2  \left\{ \sum_{t=1}^{T} (y_{it} - x_i \beta) \right\}^2 f_\theta^{1+\gamma}(y_i | x_i) dy_i\\
			& \hspace{1cm} - \frac{\sigma_\alpha^2 (2\sigma_\epsilon^2 + T \sigma_\alpha^2 )} {4\sigma_\epsilon^4 ( \sigma_\epsilon^2 + T \sigma_\alpha^2 )^4} \int_{y_i}  \left\{ \sum_{t=1}^{T} (y_{it} - x_{it} \beta) \right\}^4  f_\theta^{1+\gamma}(y_i | x_i) dy_i\\
			& \hspace{1cm} -\frac{T}{4 \sigma_\epsilon^4  ( \sigma_\epsilon^2 + T \sigma_\alpha^2 )}	\sum_{t=1}^{T}  \int_{y_i} (y_{it} - x_{it} \beta)^2 f_\theta^{1+\gamma}(y_i | x_i) dy_i\\
			& \hspace{1cm} + \Bigg[\frac{T\sigma_\alpha^2 (2\sigma_\epsilon^2 + T \sigma_\alpha^2 )} {4\sigma_\epsilon^4 ( \sigma_\epsilon^2 + T \sigma_\alpha^2 )^3} -\frac{T \left[  \sigma_\epsilon^2 + (T-1) \sigma_\alpha^2 \right]}{4 \sigma_\epsilon^{2} ( \sigma_\epsilon^2 + T \sigma_\alpha^2 )^3} \Bigg]  \int_{y_i} \left\{ \sum_{t=1}^{T} (y_{it} - x_{it} \beta) \right\}^2  f_\theta^{1+\gamma}(y_i | x_i) dy_i\\
			&=  \frac{T^2 \left[  \sigma_\epsilon^2 + (T-1) \sigma_\alpha^2 \right]}{4 \sigma_\epsilon^{2} ( \sigma_\epsilon^2 + T \sigma_\alpha^2 )^2} \times	M  (1+\gamma)   \mbox{, using \eqref{int_f}}\\
			& \hspace{1cm}  + \frac{1} {4 \sigma_\epsilon^4 ( \sigma_\epsilon^2 + T \sigma_\alpha^2 )^2} \times T M (1+\gamma)\Big[ (T+2)  \sigma_\epsilon^4 + (T^2 + 2T +3) \sigma_\epsilon^2 \sigma_\alpha^2   \\
			& \hspace{2cm} + 3(T^2 -T +1) \sigma_\alpha^4 \Big]    \mbox{, using \eqref{E_s2_s2}}\\
			& \hspace{1cm} - \frac{\sigma_\alpha^2 (2\sigma_\epsilon^2 + T \sigma_\alpha^2 )} {4\sigma_\epsilon^4 ( \sigma_\epsilon^2 + T \sigma_\alpha^2 )^4} \times  \frac{3MT^2}{(1+\gamma)} ( \sigma_\epsilon^2  + T \sigma_\alpha^2)^2    \mbox{, using \eqref{s4}}\\
			& \hspace{1cm} -\frac{T}{4 \sigma_\epsilon^4  ( \sigma_\epsilon^2 + T \sigma_\alpha^2 )}	 \times T M ( \sigma_\epsilon^2  + \sigma_\alpha^2 )      \mbox{, using \eqref{int_y2}} \\
			& \hspace{1cm} + \frac{T[T \sigma_\alpha^4 - (T-3) \sigma_\alpha^2 \sigma_\epsilon^2 - \sigma_\epsilon^4]} {4\sigma_\epsilon^4 ( \sigma_\epsilon^2 + T \sigma_\alpha^2 )^3}  \times  MT(\sigma_\epsilon^2+T\sigma_\alpha^2)     \mbox{, using \eqref{s2}}\\
			&=  \frac{T M  (1+\gamma) }{4 \sigma_\epsilon^4 ( \sigma_\epsilon^2 + T \sigma_\alpha^2 )^2}  \Big[  2(T+1)  \sigma_\epsilon^4 + (2T^2  + T +3) \sigma_\epsilon^2 \sigma_\alpha^2   + 3(T^2 -T +1) \sigma_\alpha^4 \Big] \\
			& \hspace{1cm} - \frac{3MT^2\sigma_\alpha^2 (2\sigma_\epsilon^2 + T \sigma_\alpha^2 )} {4\sigma_\epsilon^4 (1+\gamma) ( \sigma_\epsilon^2 + T \sigma_\alpha^2 )^2}  	 - \frac{M T^2  \left[  (T-1) \sigma_\alpha^2 +  \sigma_\epsilon^2 \right]} {2\sigma_\epsilon^2 ( \sigma_\epsilon^2 + T \sigma_\alpha^2 )^2} .
		\end{split}
	\end{equation}

	\subsection{Integrals for $\bJ^{(i)}$}
	
	\begin{equation}
		\begin{split}
			\int_{y_{i}} (y_{it} - & x_{it}\beta)(y_{it'} -  x_{it'}\beta)  f_\theta^{1+\gamma}(y_{i} | x_i) dy_{i} \\
			&= \int_{z_i} z_{it}z_{it'} f_\theta^{1+\gamma}(z_i | 0) dz_i \mbox{,  where } f_\theta(z_i|0) \mbox{ is } N_T(0, \bO)  \\
			&=  (2\pi)^{-\frac{T \gamma}{2}} |\bO |^{-\frac{\gamma}{2}}  \int_{z_i} z_{it}z_{it'} (2\pi)^{-\frac{T}{2}} |\bO |^{-\frac{1}{2}} \exp \left\{ -\frac{1+\gamma}{2} z_i' \bO^{-1} z_i \right\} dz_i\\
			&=  (2\pi)^{-\frac{T \gamma}{2}} |\bO |^{-\frac{\gamma}{2}} (1+\gamma)^{-\frac{T}{2}}  \int_{z_i} z_{it}z_{it'} (2\pi)^{-\frac{T}{2}} \left|\frac{\bO}{1+\gamma}\right| ^{-\frac{1}{2}} \exp \left\{ -\frac{1}{2} z_i' \left(\frac{\bO}{1+\gamma}\right)^{-1} z_i \right\} dz_i\\
			&=  (2\pi)^{-\frac{T \gamma}{2}} |\bO |^{-\frac{\gamma}{2}} (1+\gamma)^{-\frac{T+2}{2}} \bO_{tt'}\\
			& = M \bO_{tt'},
			\label{int_yy1}
		\end{split}
	\end{equation}
	where 
	\begin{equation}
		\begin{split}
			M &= (2\pi)^{-\frac{T \gamma}{2}} |\bO |^{-\frac{\gamma}{2}} (1+\gamma)^{-\frac{T+2}{2}} \\
			& =  (2\pi)^{-\frac{T \gamma}{2}}  (1+\gamma)^{-\frac{T+2}{2}} \left\{ \sigma_\epsilon^{2(T-1)} ( \sigma_\epsilon^2 + T \sigma_\alpha^2 ) \right\}^{-\frac{\gamma}{2}} \mbox{, using \eqref{inverse}}\\
			& = (2\pi)^{-\frac{T \gamma}{2}}  (1+\gamma)^{-\frac{T+2}{2}}  \sigma_\epsilon^{-\gamma(T-1)} ( \sigma_\epsilon^2 + T \sigma_\alpha^2 )^{-\frac{\gamma}{2}}. \label{M}
		\end{split}
	\end{equation} 
	For $t=t'$, combining \eqref{omega} and \eqref{int_yy1}, we get
	\begin{equation}
		\int_{y_{i}} (y_{it} -  x_{it}\beta)^2 f_\theta^{1+\gamma}(y_{i} | x_i) dy_{i} 
		= M ( \sigma_\epsilon^2  + \sigma_\alpha^2 ).
		\label{int_y2}
	\end{equation}
	For $t\neq t'$,  combining \eqref{omega} and \eqref{int_yy1}, we get
	\begin{equation}
		\int_{y_{i}} (y_{it} -  x_{it}\beta)(y_{it'}-  x_{it'}\beta) f_\theta^{1+\gamma}(y_{i} | x_i) dy_{i} 
		=    M  \sigma_\alpha^2 .
		\label{int_yy}
	\end{equation}
	
	\noindent
	For two integer $r$ and $s$, where $(r+s)$ is an odd number, we have
	\begin{equation}
		\begin{split}
			\int_{y_{i}} (y_{it} -&  x_{it}\beta)^r (y_{it} -  x_{it} \beta)^s  f_\theta^{1+\gamma}(y_{i} | x_i) dy_{i}\\
			&= \int_{z_i} z_{it}^r z_{it}^s f_\theta^{1+\gamma}(z_i | 0) dz_i \mbox{,  where } f_\theta(z_i|0) \mbox{ is } N_T(0, \bO)  \\
			&=0 \mbox{, using \eqref{s_sum_s2}}.
			\label{int_sr_ss}
		\end{split}
	\end{equation}
	
	\noindent
	Now
	\begin{equation}
		\begin{split}
			\int_{y_{i}}  f_\theta^{1+\gamma}(y_{i} | x_i) dy_{i} 
			&= \int_{z_i}  f_\theta^{1+\gamma}(z_i | 0) dz_i \mbox{,  where } f_\theta(z_i|0) \mbox{ is } N_T(0, \bO)   \\
			&=  (2\pi)^{-\frac{T \gamma}{2}} |\bO |^{-\frac{\gamma}{2}}  \int_{z_i}  (2\pi)^{-\frac{T}{2}} |\bO |^{-\frac{1}{2}} \exp \left\{ -\frac{1+\gamma}{2} z_i' \bO^{-1} z_i \right\} dz_i\\
			&=  (2\pi)^{-\frac{T \gamma}{2}} |\bO |^{-\frac{\gamma}{2}} (1+\gamma)^{-\frac{T}{2}}  \int_{z_i}  (2\pi)^{-\frac{T}{2}} \left|\frac{\bO}{1+\gamma}\right| ^{-\frac{1}{2}} \exp \left\{ -\frac{1}{2} z_i' \left(\frac{\bO}{1+\gamma}\right)^{-1} z_i \right\} dz_i\\
			&=  (2\pi)^{-\frac{T \gamma}{2}} |\bO |^{-\frac{\gamma}{2}} (1+\gamma)^{-\frac{T}{2}}\\
			& = M(1+\gamma) \mbox{, using }\eqref{M}.
			\label{int_f}
		\end{split}
	\end{equation}
	So
	\begin{equation}
		\begin{split}
			\int_{y_i} &\left\{ \sum_{t=1}^{T} (y_{it} - x_{it} \beta) \right\}^2  f_\theta^{1+\gamma}(y_i | x_i) dy_i \\
			& = \int_{z_i}  \left\{ \sum_{t=1}^T z_{it}  \right\}^2 f_\theta^{1+\gamma}(z_i | 0) dz_i \mbox{,  where } f_\theta(z_i|0) \mbox{ is } N_T(0, \bO)   \\
			&=  (2\pi)^{-\frac{T \gamma}{2}} |\bO |^{-\frac{\gamma}{2}}  \int_{z_i} \left\{ \sum_{t=1}^T z_{it}  \right\}^2 (2\pi)^{-\frac{T}{2}} |\bO |^{-\frac{1}{2}} \exp \left\{ -\frac{1+\gamma}{2} z_i' \bO^{-1} z_i \right\} dz_i\\
			&=  (2\pi)^{-\frac{T \gamma}{2}} |\bO |^{-\frac{\gamma}{2}} (1+\gamma)^{-\frac{T}{2}} \int_{z_i} \left\{ \sum_{t=1}^T z_{it}  \right\}^2 (2\pi)^{-\frac{T}{2}} \left|\frac{\bO}{1+\gamma}\right| ^{-\frac{1}{2}} \exp \left\{ -\frac{1}{2} z_i' \left(\frac{\bO}{1+\gamma}\right)^{-1} z_i \right\} dz_i\\   
			&=  M(1+\gamma) E\left(\left[\sum_{t=1}^T s_{it} \right]^2\right)  \mbox{, using \eqref{M}}\\
			&=  MT ( \sigma_\epsilon^2  + T \sigma_\alpha^2) \mbox{, using \eqref{Es2}}.
			\label{s2}
		\end{split}
	\end{equation}
	Similarly
	\begin{equation}
		\begin{split}
			\int_{y_i} \left\{ \sum_{t=1}^{T} (y_{it} - x_{it} \beta) \right\}^4  f_\theta^{1+\gamma}(y_i | x_i) dy_i 
			& =   M(1+\gamma) E\left(\left[\sum_{t=1}^T s_{it} \right]^4\right) \\
			&=  \frac{3MT^2}{(1+\gamma)} ( \sigma_\epsilon^2  + T \sigma_\alpha^2)^2 \mbox{, using \eqref{Es4}},
			\label{s4}
		\end{split}
	\end{equation}

	\begin{equation}
		\begin{split}
			\int_{y_i} \left\{ \sum_{t=1}^{T} (y_{it} - x_{it} \beta)^2 \right\}^2 f_\theta^{1+\gamma}(y_i | x_i) dy_i
			&= \int_{y_i} \left\{ \sum_{t=1}^{T} z_{it}^2 \right\}^2 f_\theta^{1+\gamma}(z_i | 0) dz_i \mbox{,  where } f_\theta(z_i|0) \mbox{ is } N_T(0, \bO)   \\
			& =  M E\left(\left[\sum_{t=1}^T s_{it}^2 \right]^2\right) \\
			&=  M ( T(T+2)  \sigma_\epsilon^4  +  2T(T+2)  \sigma_\epsilon^2\sigma_\alpha^2 + 3T^2 \sigma_\alpha^4 ),
			\label{int_square}
		\end{split}
	\end{equation}
	and
	\begin{equation}
		\begin{split}
			\int_{y_i} &(y_{it}  - x_{it} \beta)^2 \left\{ \sum_{t'=1}^{T} (y_{it'} - x_{it'} \beta) \right\}^2 f_\theta^{1+\gamma}(y_i | x_i) dy_i \\
			&= M(1+\gamma) E\Bigg(s_{it'}^2 \left[\sum_{t=1}^T s_{it} \right]^2\Bigg) \\
			&=  M(1+\gamma) \left\{ (T+2)  \sigma_\epsilon^4 + (T^2 + 2T +3) \sigma_\epsilon^2 \sigma_\alpha^2 + 3(T^2 -T + 1 ) \sigma_\alpha^4 \right\}  \mbox{, using \eqref{Es2s2}.
				\label{E_s2_s2} }
		\end{split}
	\end{equation}

	\subsection{Expectations for Integrals}
	
	Suppose $s_i \sim N_T \left(0, \frac{\bO}{1+\gamma}\right)$, then
	\begin{equation}
		\begin{split}
			V\left(\sum_{t=1}^T s_{it} \right) &=  \sum_{t=1}^T V( s_{it} ) +  \sum_{t \neq t'} cov( s_{it} s_{it'} )\\
			&=  \frac{1}{1+\gamma}\sum_{t=1}^T \bO_{tt} +  \frac{1}{1+\gamma} \sum_{t \neq t'} \bO_{tt'}\\
			& = \frac{1}{1+\gamma} \sum_{t=1}^T ( \sigma_\epsilon^2  +   \sigma_\alpha^2 )  +  \frac{1}{1+\gamma} \sum_{t \neq t'} \sigma_\alpha^2 \\
			& = \frac{T}{1+\gamma} ( \sigma_\epsilon^2  + T \sigma_\alpha^2) .
		\end{split}
	\end{equation}
	So $\sum_{t=1}^T s_{it} \sim N(0,\frac{T}{1+\gamma} ( \sigma_\epsilon^2  + T \sigma_\alpha^2)) $. Therefore
	\begin{eqnarray}
		E\left(\left[\sum_{t=1}^T s_{it} \right]^2\right) &=& \frac{T}{1+\gamma} ( \sigma_\epsilon^2  + T \sigma_\alpha^2),\label{Es2}\\
		E\left(\left[\sum_{t=1}^T s_{it} \right]^4\right) &=& \frac{3T^2}{(1+\gamma)^2} ( \sigma_\epsilon^2  + T \sigma_\alpha^2)^2, \label{Es4}
	\end{eqnarray}
	and
	\begin{equation}
		\begin{split}
			E\left(\left[  \sum_{t=1}^T s_{it}^2 \right]^2 \right) &=  E\left(  \sum_{t=1}^T s_{it}^4 + \sum_{t \neq t'} s_{it}^2 s_{it'}^2 \right)
			\\
			&= \sum_{t=1}^T 3 \bO^2_{tt} + \sum_{t \neq t'} ( \bO_{tt} \bO_{t't'} + 2 \bO_{tt'}^2)\\
			&=  3T \bO^2_{tt} + T(T-1) ( \bO_{tt}^2 + 2 \bO_{tt'}^2)\\
			&=   T(T+2) ( \sigma_\epsilon^2  +  \sigma_\alpha^2)^2 + 2T(T-1)  \bO_{tt'}^2\\
			&=   T(T+2)  \sigma_\epsilon^4  +  2T(T+2)  \sigma_\epsilon^2\sigma_\alpha^2 + 3T^2 \sigma_\alpha^4.
			\label{e_s2_2}
		\end{split}
	\end{equation}
	
	\noindent
	For $t' = 1, 2, \cdots, T$, we have
	\begin{equation}
		\begin{split}    
			E\Bigg(s_{it'}^2 &\left[\sum_{t=1}^T s_{it} \right]^2\Bigg) 
			=  E\left(s_{it'}^4 + s_{it'}^2  \sum_{t \neq t'} s_{it}^2 
			+s_{it'}^3  \sum_{t \neq t'} s_{it} + s_{it'}^2  \sum_{t \neq t' \neq t''} s_{it} s_{it''} \right) \\
			&=  3 \bO_{t't'}^2 + \left(\bO_{t't'} \sum_{t \neq t'} \bO_{tt} + 2 \sum_{t \neq t'} \bO_{tt'}^2 \right)
			+ 3 \bO_{t't'} \sum_{t \neq t'} \bO_{tt'} + \left(\bO_{t't'}  \sum_{t \neq t' \neq t''} \bO_{tt''}  + 2 \sum_{t \neq t' \neq t''} \bO_{tt'} \bO_{t't''}  \right)\\
			&=  3 \bO_{t't'}^2 + (T-1) \bO_{t't'}^2  + 2 (T-1) \bO_{tt'}^2 + 3(T-1) \bO_{t't'} \bO_{tt'}\\
			& \hspace{2cm}
			+ (T-1) (T-2)(\bO_{t't'} \bO_{tt'} + 2\bO_{tt'}^2)  \mbox{ for } t\neq t'\\
			&=  (T+2) ( \sigma_\epsilon^2  +  \sigma_\alpha^2)^2 + 2 (T-1) \sigma_\alpha^4 + 3 (T-1) \sigma_\alpha^2  ( \sigma_\epsilon^2  +  \sigma_\alpha^2) \\
			& \hspace{2cm}
			+ (T-1) (T-2)(\sigma_\epsilon^2  \sigma_\alpha^2 + \sigma_\alpha^4 + 2\sigma_\alpha^4),   \mbox{ using \eqref{omega}}\\
			&=  (T+2)  \sigma_\epsilon^4 + (T^2 + 2T +3) \sigma_\epsilon^2 \sigma_\alpha^2 +  3(T^2 -T + 1 ) \sigma_\alpha^4 .
			\label{Es2s2}
		\end{split}
	\end{equation}
	
	\noindent
	Similarly
	\begin{equation}
		\begin{split}    
			E\Bigg(s_{it'} &\left[\sum_{t=1}^T s_{it} \right]^2\Bigg) 
			=  E\left(s_{it'}^3 + s_{it'}^2  \sum_{t \neq t'} s_{it} 
			+ s_{it'}  \sum_{t \neq t' \neq t''} s_{it} s_{it''} \right) = 0.
			\label{s_sum_s2}
		\end{split}
	\end{equation}

	\subsection{Oman Weather Stations} \label{sec:weather_station}
	
	Oman weather dataset consists observations from 55 stations across Oman over the period January 2018 to December 2018. The list of weather stations is given in Table \ref{tab:stations}.
	
	\begin{table}[htbp] \small
		\centering
		\tabcolsep 0.1in
		\caption{Stations for the monthly Oman weather data.}
		\begin{tabular}{@{}llllll@{}}
			\toprule
			Station & Station & Station & Station & Station \\
			\midrule
			Adam			& Diba           & Madha 	                     & Nizwa              & Shalim \\
			Al Amrat        & Fahud Airport & Mahdah   	                 & Qairoon Hairiti    & Shinas \\
			Al Jazir		& Haima	         & Majis 	                     & Qalhat             & Sohar Airport \\
			Al Khaboura		& Ibra			 & Marmul Airport 		         & Qarn alam          & Sunaynah  \\
			Al Mudhaibi     & Ibri	         & Masirah                       & Qurayyat           & Suwaiq \\
			Al Qabil 	    & Ibri New 		 & Mina Salalah 	             & RasAlHaad          & Taqah \\
			Al-Buraymi      & Izki  	     & Mina Sultan Qaboos            & Sadah              & Thamrayt \\
			Bidiyah 	    & Joba	         & Mirbat			             & Saham              & Wadi Bani Khalid \\          
			Buhla 			& Khasab Airport & Muqshin			             & Saiq Airport      & Yalooni \\                                  			   
			Bukha 			& Khasab Port	 & Muscat City			         & Salalah Airport   & Yalooni Airport \\                                  			
			Dhank (Qumaira) & Liwa      	 & Muscat International Airport	 & Samail             & Yanqul \\                                  			                             
			\bottomrule
		\end{tabular}
		\label{tab:stations}
	\end{table}

\end{document}